\documentclass[11pt]{article}
\usepackage[utf8]{inputenc}
\usepackage[T1]{fontenc}
\usepackage{microtype} 

\usepackage{lmodern}

\usepackage[margin=1in]{geometry}
\usepackage{amssymb,amsmath,amsfonts,amsthm,dsfont}
\usepackage{mathtools}

\usepackage[linesnumbered,boxed]{algorithm2e}

\SetAlgorithmName{Algorithm}{Algo}{List of Algos}

\newtheorem{theorem}{Theorem}[section]\numberwithin{equation}{section}
\newtheorem{corollary}[theorem]{Corollary}
\newtheorem{definition}[theorem]{Definition}

\newtheorem{fact}[theorem]{Fact}
\newtheorem{claim}[theorem]{Claim}
\newtheorem{lemma}[theorem]{Lemma}

\usepackage{caption}
\usepackage{subcaption}
\usepackage{thm-restate}

\usepackage{tikz}
\usepackage{graphicx}
\usepackage{xcolor}
\newcommand{\eps}{\varepsilon}
\newcommand{\wh}{\widehat}
\newcommand{\wt}{\widetilde}
\newcommand{\Ot}{\wt{O}}
\newcommand{\Omegat}{\wt{\Omega}}
\newcommand{\Thetat}{\wt{\Theta}}
\newcommand{\dkl}{\mathrm{KL}}
\newcommand{\weight}{\text{wt}}
\newcommand{\abs}[1]{|#1|}
\DeclareMathOperator*{\E}{\mathbb{E}}
\DeclareMathOperator*{\argmax}{arg\,max}
\DeclareMathOperator*{\argmin}{arg\,min}
\DeclareMathOperator{\Var}{Var}
\newcommand{\Psymb}{{\bf Pr}}
\DeclareMathOperator*{\ProbOp}{\Psymb}
\renewcommand{\Pr}{\ProbOp}

\newcommand{\ignore}[1]{}

\usepackage{todonotes}

\newcounter{mynotes}
\newcommand{\eric}[1]{\addtocounter{mynotes}{1}{{}}\todo[color=blue!20!white]{[\arabic{mynotes}] \scriptsize  {Eric: {\sf {#1}}}}}
\newcommand{\arnab}[1]{\addtocounter{mynotes}{1}{{}}\todo[color=blue!20!white]{[\arabic{mynotes}] \scriptsize  {Arnab: {\sf {#1}}}}}
\newcommand{\sutanu}[1]{\addtocounter{mynotes}{1}{{}}\todo[color=blue!20!white]{[\arabic{mynotes}] \scriptsize  {Sutanu: {\sf {#1}}}}}
\newcommand{\vinod}[1]{\addtocounter{mynotes}{1}{{}}\todo[color=blue!20!white]{[\arabic{mynotes}] \scriptsize  {Vinod: {\sf {#1}}}}}

\DeclarePairedDelimiterX{\infdivx}[2]{(}{)}{%
  #1\;\delimsize\|\;#2%
}
\newcommand{\kl}{D\infdivx}

\allowdisplaybreaks

\usepackage[symbol]{footmisc}

\begin{document}

\title{Near-Optimal Learning of Tree-Structured Distributions by Chow-Liu}

\author{
Arnab Bhattacharyya\\
  National University of Singapore\\
arnabb@nus.edu.sg
\and
Sutanu Gayen\\
National University of Singapore\\
 sutanugayen@gmail.com
\and
Eric Price\\
University of Texas at Austin\\
ecprice@cs.utexas.edu
\and
N. V. Vinodchandran\\
University of Nebraska-Lincoln\\
vinod@cse.unl.edu
}

\date{}

 \renewcommand{\eric}[1]{}
 \renewcommand{\arnab}[1]{}
 \renewcommand{\sutanu}[1]{}
 \renewcommand{\vinod}[1]{}

\maketitle

\begin{abstract}
  We provide finite sample guarantees for the classical Chow-Liu algorithm (IEEE Trans.~Inform.~Theory, 1968) to learn a tree-structured graphical model of a distribution. For a distribution $P$ on $\Sigma^n$ and a tree $T$ on $n$ nodes, we say $T$ is an $\varepsilon$-approximate tree for $P$ if there is a $T$-structured distribution $Q$ such that $D(P\;||\;Q)$ is at most $\varepsilon$ more than the best possible tree-structured distribution for $P$. We show that if $P$ itself is tree-structured, then the Chow-Liu algorithm with the plug-in estimator for mutual information with $\wt{O}(|\Sigma|^3 n\varepsilon^{-1})$ i.i.d.~samples outputs an $\varepsilon$-approximate tree for $P$ with constant probability. In contrast, for a general $P$ (which may not be tree-structured), $\Omega(n^2\varepsilon^{-2})$ samples are necessary to find an $\varepsilon$-approximate tree. Our upper bound is based on a new conditional independence tester that addresses an open problem posed by Canonne, Diakonikolas, Kane, and Stewart~(STOC, 2018): we prove that for three random variables $X,Y,Z$ each over $\Sigma$,  testing if $I(X; Y \mid Z)$ is $0$ or $\geq \varepsilon$ is possible with $\widetilde{O}(|\Sigma|^3/\varepsilon)$ samples. Finally, we show that for a specific tree $T$, with $\widetilde{O}(|\Sigma|^2n\varepsilon^{-1})$ samples from a distribution $P$ over $\Sigma^n$, one can efficiently learn the closest $T$-structured distribution in KL divergence by applying the add-1 estimator at each node.
\end{abstract}

\sutanu{should start with a title page}

\section{Introduction}

\ignore{
\begin{enumerate}
\item Approximate Structure Recovery with Chow Liu
\begin{enumerate}
\item Nonrealizable
\item Realizable
\end{enumerate}
\item Lower Bounds on  Sample Complexity 
\begin{enumerate}
\item Nonrealizable
\item Realizable
\end{enumerate}
\item Learning the distribution when the structure is known
\end{enumerate}
}

Probabilistic graphical models form a highly effective framework for encoding high-dimensional distributions. Graphical models yield human-interpretable representation of data as they explicitly describe the statistical dependencies among  different features. From a computational standpoint, the graphical representation enables efficient algorithms for inference, e.g., message passing, loopy belief propagation, and other variational inference methods \cite{KFL01}. Graphical models have found extensive applications in many domains, such as image processing, natural language processing and computational biology; see \cite{Lau96, KF09, WJ08} and the references therein for examples.

A fundamental question in this area is to learn graphical models from independently drawn samples. In this paper, we focus on the basic problem of learning {\em tree-structured distributions}. Given a tree $T$ on $n$ nodes, fix an arbitrary root and orient it outwards. A distribution $P$ over variables $X_1, \dots, X_n$  is said to be {\em $T$-structured} iff for every non-root vertex $i$:
\[
X_ i = f_i(X_{\text{pa}(i)}, U_i)
\]
where $\text{pa}(i)$ is the parent of $i$ in the oriented tree, $U_i$ is an independent random variable, and $f_i$ is a (deterministic) function. A distribution is tree-structured if it is $T$-rooted for some tree $T$. Equivalently, a tree-structured distribution is a Markov random field where the underlying undirected graph is a tree. 

In a seminal work \cite{DBLP:journals/tit/ChowL68}, Chow and Liu observed that the tree-structured distribution maximizing the likelihood of the observed samples can be obtained by solving a maximum weight spanning tree problem. In particular, their algorithm assigns a weight equal to the empirical mutual information between each pair of variables and finds a maximum weight spanning tree in this weighted graph. The resulting tree can be oriented from an arbitrary root, so as to assign a parent $\text{pa}(i)$ for all non-root vertices $i$. Finally, the conditional probability distributions $X_i \mid X_{\text{pa}(i)}$ can be learned from the data.

Chow and Wagner \cite{CW73} showed that the Chow-Liu algorithm consistently recovers structure, meaning that if the samples are generated by a $T$-structured distribution for a tree $T$, then it recovers $T$ with probability approaching $1$ as the number of samples tends to infinity.  More recent works \cite{TATW11, TTZ20} have used large-deviation theory to study the error exponent $K_P$
of $T$-structured distributions $P$, where:
\[
K_P = \lim_{N \to \infty} -\frac1N \log \Pr[\wh{T} \neq T].
\]
Here, $\wh{T}$ is the tree output by the Chow-Liu algorithm, and $N$
is the number of samples used.  The bounds they obtain depend on the
distribution structure, since it may be very hard to distinguish $T$
from an alternative tree that is almost as good.\eric{Can we be more specific in comparing their bounds to ours?}

In this work, we take a different viewpoint that is in the spirit of distribution learning and probabilistically approximately correct (PAC) analysis \cite{Val84, KSS94, KS94}. Instead of trying to exactly recover the structure of a tree-structured distribution $P$, we consider the objective of {\em learning a tree-structured distribution $Q$ that is close to $P$}. For many downstream tasks, most notably {statistical inference}, it is fine to not recover the exact structure as long as one can approximate probabilities of relevant events using the learned distribution. Also, this viewpoint allows us to analyze  Chow-Liu for non-tree structured distributions $P$, by comparing how far $P$ is from the output of Chow-Liu and how far from the closest tree-structured distribution.

More formally, for a distribution $P$ over $\Sigma^n$ and a tree $T$ on $n$ vertices, let:
\[
P_T := \argmin_{\substack{T\text{-structured}\\ \text{distribution } Q}} \kl{P}{Q}
\]
where $\kl{\cdot}{\cdot}$ denotes the KL-divergence. We say that a tree $\wh{T}$ is an {\em $\eps$-approximate tree} for $P$ if:
\[
\kl{P}{P_{\wh{T}}} \leq \min_{\text{tree }T} \kl{P}{P_{T}} + \eps.
\]
The KL divergence, although not a metric, is a useful notion of distance to consider in this setting. Firstly, with infinite samples, Chow-Liu's output maximizes the likelihood of generating samples from $P$ and hence, minimizes $\kl{P}{\cdot}$. Secondly, via Pinsker's inequality, bounding the KL divergence by $\eps$ implies a $\sqrt{2\eps}$ bound on total variation distance which may be more directly useful. 

\subsection{Our Contributions}
We study the number of samples required by Chow-Liu to output an $\eps$-approximate tree with a fixed error probability. We first observe that for any distribution $P$, it can be guaranteed that the output of Chow-Liu is $\eps$-approximate if each mutual information estimate is an additive $\pm \frac{\eps}{2n}$ estimate. Known bounds for the plug-in entropy estimator imply the following sample complexity.

\ignore{
\begin{definition}[Realizable Tree Structure Learning Problem] Given $\eps,\delta <1$ and samples from an unknown distribution $P$ on $\Sigma^n$ over an unknown tree $G$ on $n$ nodes. Output a tree $H$ so that the best distribution $Q$ on $\Sigma^n$ over $H$ satisfies that $\kl{P}{Q} \leq \eps$ with probability at least $1-\delta$. 
\end{definition}}

\begin{lemma}\label{lem:mainnr}
The \emph{\textsf{Chow-Liu}} algorithm when run on $\widetilde{O}\left(\frac{|\Sigma|^2 n}{\eps} \right.+\allowbreak \left.\frac{n^2}{\eps^{2}}\log {\frac1 \delta}\right)$ samples from a distribution $P$ on $\Sigma^n$ outputs an $\eps$-approx. tree $T$ with probability at least $1-\delta$. Moreover, the dependence of the sample complexity on $n$ and $\eps$ are tight up to logarithmic factors.
\end{lemma}

We show that the quadratic dependence on $n$ and $\eps$ is
unfortunately necessary for general distributions $P$.  However, in
the ``realizable'' setting where $P$ is actually tree-structured, we
show that the sample complexity can be improved to near-linear:
 
\begin{theorem}\label{thm:mainr}
The \emph{\textsf{Chow-Liu}} algorithm when run on $\widetilde{O}({|\Sigma|^3 n\over \eps}\log {1\over \delta})$ samples from a tree-structured distribution $P$ on $\Sigma^n$ outputs an $\eps$-approximate tree $T$ with probability at least $1-\delta$. Moreover, the dependence on $n$ and $\eps$ are tight up to logarithmic factors.
\end{theorem}
Hence, for example, for tree-structured Ising models (where $\Sigma = \{\pm 1\}$), there is a provable near-quadratic gap in the sample complexity for realizable versus non-realizable input distributions. Note that with $O(n/\eps)$ samples, we do not get accurate estimates of the mutual information edge weights. However, as we explain in Section \ref{sec:overview}, 
the errors for the edge weights are not independent; in fact, the errors are correlated in such a way that Chow-Liu still recovers an approximate tree!
We note that our $\Omega(n/\eps)$-sample complexity lower bound is specifically for \emph{recovering the structure} of the unknown tree. Daskalakis, Dikkala, and Kamath~\cite{DBLP:journals/tit/DaskalakisD019} have shown the same lower bound for learning the distribution, but learning the tree might have been easier.

Our main tool for proving Theorem \ref{thm:mainr} is a result on testing conditional independence using the plug-in conditional mutual information estimator. We show that $\widetilde{O}(|\Sigma|^3/\eps)$ samples suffice to distinguish $I(X;Y \mid Z)=0$ from $I(X;Y\mid Z) \geq \eps$ with constant probability. In more detail:
\begin{restatable}[Conditional Mutual Information Tester]{theorem}{condintest}
\label{thm:conditionalindependencetesting}
  Let $(X, Y, Z)$ be three random variables over $\Sigma$, and
  $(\wh{X}, \wh{Y}, \wh{Z})$ be the empirical distribution over a size
  $N$ sample of $(X, Y, Z)$. There exists a universal constant $0<C<1$ so that for any
  \[
    N\ge \Theta\left({|\Sigma|^3\over \eps}\log {\abs{\Sigma} \over \delta}\log {\abs{\Sigma} \log(1/\delta)\over \eps}\right),
  \]
  the following results hold with probability $1-\delta$:
  \begin{enumerate}
  \item If $I(X; Y \mid Z) = 0$, then $I(\wh{X}; \wh{Y} \mid \wh{Z}) < \eps$.
  \item If $I(X; Y \mid Z) \ge \eps$, then $I(\wh{X}; \wh{Y} \mid \wh{Z}) > C\cdot I(X; Y \mid Z)$.
  \end{enumerate}
\end{restatable}
We also get a similar result for unconditional mutual information
testing (testing if $I(X; Y) = 0$ or $I(X; Y) \geq \eps$) with a
$\abs{\Sigma}$ factor smaller $N$.  Conditional independence testing
has previously been studied in~\cite{canonne2018testing}, which gave
optimal bounds for testing whether $(X, Y, Z)$ is conditionally
independent or $\eps$-far from independent in \emph{total variation
  distance}.  Developing a (conditional) independence tester with respect to
\emph{mutual information} with $o(\frac{1}{\eps^2})$
sample complexity was posed as an open problem
in~\cite{canonne2018testing};
Theorem~\ref{thm:conditionalindependencetesting} resolves this with
optimal $\eps$ dependence.  Moreover, the test statistic used by
Theorem~\ref{thm:conditionalindependencetesting} is simply the
empirical mutual information, which is key for our application to
Chow-Liu.

Theorem~\ref{thm:mainr} describes how Chow-Liu finds a good tree $T$.
Our final result shows how to estimate the nearest $T$-structured
distribution for fixed $T$. As above, the spirit of our approach is to make the algorithms as simple as possible (moving possible complications to the analysis). For the fixed-structure learning problem, the most natural approach is to empirically estimate $X_i \mid X_{\text{pa}(i)}=x$ for each non-root node $i$ and for each setting $x$ of the parent of $i$. However, for KL divergence, the empirical estimator is known to not work; so, we move to the next most natural estimator: Laplace's add-1 estimator \cite{Lap95}.
\vinod{The above para is confusing to me.  Eric: Is it still confusing?}

\begin{restatable}[]{theorem}{fixedlearning}\label{thm:fixed}
  Let $P$ be a discrete distribution over $\Sigma^n$.  Let $T$ be a tree on $n$ vertices, and $Q$ be a $T$-structured
  distribution with conditional probabilities at each node estimated using the empirical add-1
  estimator on \[N=\Theta\left(\frac{n\abs{\Sigma}^2}{\eps}\log {n\abs{\Sigma}\over \delta}\log \left( \frac{n\abs{\Sigma}}{\eps}\log {1\over \delta} \right)\right)\] samples from $P$.  Then\sutanu{this way of writing should appeal more to CS readers}
  \[
    \kl{P}{Q} - \kl{P}{P_T} \le \eps
  \]
  with probability $1-\delta$.
\end{restatable}
The result actually holds for arbitrary Bayes net models, not just trees.  The sample complexity becomes $\widetilde{O}(n|\Sigma|^{d+1}/\eps)$ for Bayes nets with in-degree at most $d$.  To the best of our knowledge, this is the first efficient algorithm with this guarantee.
Combining Theorem~\ref{thm:fixed} with Theorem~\ref{thm:mainr} shows
that, for any tree-structured distribution $P$, after
$\Ot(\abs{\Sigma}^3 n/\eps)$ samples, we can properly learn a tree-structured distribution $Q$
satisfying $\kl{P}{Q} \leq \eps$ (and hence
$\|P-Q\|_{TV} \leq \sqrt{2\eps}$) (see Algorithm~\ref{algo:algoOverall}).
\begin{algorithm}
\SetKwInOut{Input}{input}
\SetKwInOut{Output}{output}
\SetAlgoLined
\caption{Learning Tree-structured Distributions from Samples}
\label{algo:algoOverall}
\Input{Sample access to $P$ over $X_1,\dots,X_n\in \Sigma$}
\Output{A tree Bayes net $Q$}
\BlankLine
$T\gets$ MST of empirical mutual informations (Lemma~\ref{lem:mainnr} or Theorem~\ref{thm:mainr}) \tcp*{structure learning}
$Q\gets$ the Laplace estimator over the edges of $T$ (Theorem~\ref{thm:fixed})\tcp*{conditional probabilities}
\KwRet{$Q$}
\end{algorithm}


\subsection{Related Work}
\ignore{
Learning a multivariate distribution from samples
is an important problem in machine
learning and statistics with many applications. 
The problem is provably intractable for general high-dimensional multivariate distributions (\textcolor{red}{what to cite?}. Thus, structural assumptions are to be made for designing efficient and practical learning algorithms for high-dimensional distributions. Graphical models including Bayesian networks and Markov Random Fields (MRFs)  
are widely popular natural classes of structured distributions. In this setup, the problem natural decomposes into two subproblems: {\em structure learning} and {\em parameter learning}.

For structure learning, the goal is to output the best structure (e.g.: a Bayes net that maximizes the likelihood of the data) as well as the parameters of the distributions, given independent samples. Unfortunately, in general for both Bayes nets and MRFs, finding the best structure is known to be NP-hard~\cite{DBLP:conf/soda/KargerS01,DBLP:conf/aistats/Chickering95}.  In this context, Chow-Liu algorithm remains one of the few efficient structure learning algorithms. Since its publication, researchers continues to  look into analyzing properties of this algorithm \cite{CW73,DBLP:conf/icml/Meila99,TATW11, TTZ20} and generalizing it to other classes of graphs, e.g., polytrees \cite{Dasgupta13}, bounded treewidth graphs \cite{Sre03, NB04}, and mixtures of trees~\cite{MJ01, AHHK12}. Most of the earlier analysis focus on proving properties of the algorithm, most notably when the algorithm recovers the exact tree structure, in the limit that the number of samples tends to infinity. 

In this work we consider a PAC-style analysis of Chow-Liu algorithm. The only previous 
 
One of the main technical components of 
our analysis of Chow-Liu algorithm is a new
{\em conditional independence tester} which may be of independent interest in the area of {\em distribution property testing}~\cite{GoldreichR11,DBLP:conf/focs/BatuFFKRW01}.  
Distribution property testing is vast and rapidly progressing area with many significant results. We refer the reader to the surveys \cite{Canonne15Survey,DBLP:journals/crossroads/Rubinfeld12} and the textbook \cite{DBLP:books/cu/Goldreich17} and references therein for more details and results on the topic. Testing independence of two or more random variables has received some attention in distribution testing~\cite{DBLP:conf/focs/BatuFFKRW01,DBLP:conf/nips/AcharyaDK15,canonne2018testing}. Simplest formulation of the problem is the following: Given samples from an unknown joint distribution on variables $(X,Y)$: decide with probability $\geq 2/3$ whether $X$ and $Y$ are independent or they are $\eps$ far (under some distance measure) from the product distribution $X\times Y$. Very recently, motivated by its practical importance, Canonne et. al. considered the problem of conditional independence testing~\cite{canonne2018testing}. In particular, given samples from an unknown discrete random variable $(X, Y, Z)$ on domain $[\ell_1]\times [\ell_2] \times [n]$, distinguish, with probability at least 2/3, between the case that $X$ and $Y$ are conditionally independent given $Z$ from the case that $(X,Y,Z)$ is $\eps$-far, from every distribution that has this property. They establish optimal sample complexity bound for the case when $l_1$ and $l_2$ are constants as a function of $n$ and $\eps$. The key difference in our setting is that we are interested in the case when all variables takes values from the same alphabet $\Sigma$. 

Finally, we would like to point out that even though structure learning of Bayes nets is computationally hard, there are efficient algorithms known for parameter learning of known Bayes nets~\cite{DBLP:journals/ml/Dasgupta97,DBLP:journals/corr/abs-2002-05378}.

\medskip

\noindent{\color{red} Below refs. not  yet discussed}

General Complexity Sample complexity bounds:
\cite{DBLP:conf/sigecom/Brustle0D20}

Learning structure of Bayes nets (under practical assumptions): at lease two lines of work~\cite{DBLP:conf/uai/Squires0U20} and \cite{DBLP:journals/corr/abs-2006-11970}, and also \cite{DBLP:conf/sigecom/Brustle0D20}.

Learning and testing high-dimensional distributions:
\cite{DBLP:journals/tit/DaskalakisD019,DBLP:journals/tit/CanonneDKS20,DBLP:conf/colt/DaskalakisP17,DBLP:conf/focs/KlivansM17}
}

Learning a multivariate distribution from samples
is an important problem in machine
learning and statistics with many applications. 
The problem is provably intractable for general high-dimensional multivariate distributions (e.g., \cite{pmlr-v40-Kamath15}). Thus, structural assumptions need to be made for designing efficient and practical learning algorithms for high-dimensional distributions. Graphical models including Bayesian networks and Markov Random Fields (MRFs)  
are widely popular natural classes of structured distributions. In this setup, the learning problem naturally decomposes into two subproblems: {\em structure learning} and {\em parameter learning}.

For structure learning, the goal is to output the best structure (eg: a Bayes net that maximizes the likelihood of the data), given independent samples. Unfortunately, in general for both Bayes nets and MRFs, finding the best structure is known to be NP-hard~\cite{DBLP:conf/aistats/Chickering95, DL97, DBLP:conf/soda/KargerS01, Meek01}. 
In this context, Chow-Liu algorithm remains one of the few efficient structure learning algorithms that does not require any additional assumptions.  Since its publication, researchers have continued to  look into analyzing properties of this algorithm \cite{CW73,DBLP:conf/icml/Meila99,TATW11, TTZ20} and generalizing it to other classes of graphs, e.g., polytrees \cite{Dasgupta13}, bounded treewidth graphs \cite{Sre03, NB04}, and mixtures of trees~\cite{MJ00, AHHK12}. Most of these works focus on establishing conditions guaranteeing that the algorithm recovers the {\em exact} tree structure in the limit that the number of samples tends to infinity.  Also, for general graph-structured Ising and Markov random fields, several algorithms \cite{BMS13, WRN13, Bres15, DBLP:conf/focs/KlivansM17, WSD19, Goel20} have been proposed that recover the graphical structure under various distributional assumptions. 

As mentioned in the introduction, a common motivation for structure learning is to subsequently use the structure for inference algorithms. For such applications, instead of recovering the exact structure, it is more relevant to recover a model that approximates the original distribution statistically and on which inference can be performed efficiently. For example, Wainwright \cite{Wai06} discusses situations in which it is computationally beneficial to use inconsistent learning algorithms (even in the infinite sample limit) to feed into approximate inference algorithms.  Trees play an important role for inference algorithms, since the commonly used sum-product algorithm assumes tree structure, and other more general inference algorithms (like the junction tree algorithm and various approximate inference techniques) rely on tree-like structure.  This is what motivates the notion of learning $\eps$-approximate trees considered in this paper.

The problem of learning $\eps$-approximate graphical models has a long history. H\"offgen \cite{Hof93} studied the $\eps$-approximate structure learning of an unknown Bayes net over $\{0,1\}^n$ of indegree $d$ as a combinatorial optimization problem and gave a sample complexity of $\widetilde{O}(n^24^d\eps^{-2})$\footnote{H\"offgen's capped the empirical probabilities away from 0 and 1 and then used a plug-in estimator for entropy/MI.}. He showed that the optimization problem is efficient for trees ($d=1$), essentially establishing Lemma \ref{lem:mainnr} for distributions on $\{0,1\}^n$.
There have been several other works which provide PAC-learning guarantees for generalizations of trees: 
bounded tree-width junction trees \cite{NB04, CG08}, factor graphs \cite{AKN06}, and forest-structured MRF's \cite{LXGGLW11}. While we consider the KL divergence between the true distribution and the output of Chow-Liu. Bresler and Karzand \cite{bresler2020} recently studied the same question with respect to maximum total variation distance between pairwise marginals. Their work is restricted to Ising models, and their sample complexity depends on bounds on distributional parameters (edge weights) while ours do not.  In another recent work, Brustle, Cai and Daskalakis \cite{BCD20} (generalizing the results in \cite{DMR20}) get bounds on the sample complexity of learning $\eps$-approximate\footnote{In total variation distance rather than KL} MRF's with bounded hyper-edges and Bayesian networks with bounded in-degree, but they do not get efficient algorithms for these problems.

The main technical component of 
our analysis of the Chow-Liu algorithm is a new
{\em conditional independence tester} which falls in the framework of {\em distribution property testing}~\cite{GoldreichR11,DBLP:conf/focs/BatuFFKRW01}.
We refer the reader to the surveys \cite{Canonne15Survey,DBLP:journals/crossroads/Rubinfeld12} and the textbook \cite{DBLP:books/cu/Goldreich17} and references therein for more details and results in this rapidly progressing field. Testing independence of two or more random variables has received some attention in distribution
testing~\cite{DBLP:conf/focs/BatuFFKRW01,DBLP:conf/nips/AcharyaDK15,DBLP:conf/focs/DiakonikolasK16,canonne2018testing,diakonikolas2020optimal}. The simplest formulation of the problem is the following: Given samples from an unknown joint distribution on variables $(X,Y)$: decide with probability $\geq 2/3$ whether $X$ and $Y$ are independent or they are $\eps$ far (under some distance measure) from the product distribution $X\times Y$. Very recently, motivated by its practical importance, \cite{canonne2018testing} considered the problem of conditional independence testing. In particular, given samples from an unknown discrete random variable $(X, Y, Z)$ on domain $[\ell_1]\times [\ell_2] \times [n]$, distinguish, with probability at least 2/3, between the case that $X$ and $Y$ are conditionally independent given $Z$ from the case that $(X,Y,Z)$ is $\eps$-far in TV distance from every distribution that has this property. 
The key difference in our setting is that we are interested in the stronger notion of KL divergence. 

The parameter learning problem (i.e., learning the distribution with given structure) is also well-studied. 
Dasgupta \cite{DBLP:journals/ml/Dasgupta97} showed an $\wt{O}(n^2 2^d \log(1/\delta)/\eps^2)$ for learning an $\eps$-approximate Bayes net on $n$ boolean variables with in-degree at most $d$. We improve the dependence on $n$ and $\eps^{-1}$ to linear. For the realizable setting\footnote{This result (for TV distance) was also claimed in the appendix of~\cite{DBLP:journals/tit/CanonneDKS20}, but the analysis there appears to incomplete \cite{Can20}.}, \cite{DBLP:journals/corr/abs-2002-05378} also obtained the same improvement but for constant error probability $\delta$. The key to obtaining our $\log(1/\delta)$ dependence is a PAC analysis of the add-1 estimator, which is new to the best of our knowledge. Kamath, Orlitsky, Pichapati and Suresh~\cite{pmlr-v40-Kamath15} analyze the expected risk, which does not directly imply a high-probability bound. Recently, Bresler and Karzand~\cite{9174341} studied the parameter recovery problem for tree Ising models with respect to the TV distance.

Finally, we note that while our work focuses on learning tree structured distributions, recent works \cite{DBLP:journals/tit/DaskalakisD019,DBLP:journals/tit/CanonneDKS20,DBLP:conf/colt/DaskalakisP17, DBLP:journals/corr/abs-2002-05378} have investigated testing problems for more general classes of high-dimensional distributions.

\paragraph{Concurrent Work.}
During preparation of this paper, a concurrent work by Daskalakis and
Pan~\cite{DBLP:journals/corr/abs-2010-14864} was posted online.  The headline result---that
Chow-Liu learns tree-structured distributions with near-linear number
of samples---is the same.  The techniques employed are quite
different and more involved, with~\cite{DBLP:journals/corr/abs-2010-14864} working in
squared Hellinger distance rather than KL and not involving the
connection to conditional independence testing
(Theorem~\ref{thm:conditionalindependencetesting}).  The details of
the theorem are also somewhat different, most notably in that our
result uses a $\log n$ factor more samples
while~\cite{DBLP:journals/corr/abs-2010-14864} only works for a binary
alphabet $\Sigma$.


\section{Proof Overview}\label{sec:overview}
For the purposes of this proof overview, we consider a constant size alphabet
$\Sigma$.

\subsection{Finding an Approximate Tree}
For any distribution $P$ and a tree $T$, it is known that $P_T$ is simply the
distribution that matches the marginals of $P$ on each edge of $T$.
The Chow-Liu algorithm~\cite{DBLP:journals/tit/ChowL68} is based on the following observation:
\begin{align}\label{eq:ChowLiuEquality}
  \kl{P}{P_T} = J_P - \text{wt}_P(T)
\end{align}
where $J_P = \sum_v H(P_v) - H(P)$ is independent of $T$ ($P_v$ is the marginal on variable $v$), and
\[
  \text{wt}_P(T) := \sum_{(X, Y) \in T} I(X; Y)
\]
is the weight of $T$ in the complete graph weighted by pairwise mutual
information.  Therefore $\kl{P}{P_T}$ is minimized when $T$ is the
maximum weight spanning tree $T^*$ of this weighted complete graph.

The main question is how many samples are necessary for the
maximum-weight spanning tree of the \emph{empirical} distribution
$\wh{P}$ to have nearly-optimal weight under the true distribution
$P$.  That is, for Chow-Liu to recover a $\wh{T}$ with
$\kl{P}{P_{\wh{T}}} \leq \kl{P}{P_{T^*}}  + \eps$, it is necessary and sufficient
that
\[
  \wh{T} = \argmax \text{wt}_{\wh{P}}(T)
\]
satisfies
\begin{align}\label{eq:goal}
  \text{wt}_{P}(\wh{T}) \geq \text{wt}_{P}(T^*)  - \eps.
\end{align}

\paragraph{The non-realizable setting.}  The simplest approach to
achieving~\eqref{eq:goal} would be to ensure that
$\abs{I(\wh{X}; \wh{Y}) - I(X; Y)} \leq \frac{\eps}{2n}$ for all
vertex pairs $(X, Y)$.  This guarantees for every $T$ that
$\abs{\text{wt}_{\wh{P}}(T) - \text{wt}_{P}(T)} \leq \abs{T}
\frac{\eps}{2n} < \eps/2$, which gives~\eqref{eq:goal}.  Estimating
mutual information to within $\frac{\eps}{2n}$ is possible with
$\Thetat(n^2/\eps^2)$ samples, with high probability.  A union bound
over all vertex pairs then gives the Lemma~\ref{lem:mainnr} upper bound.
We also show that this bound is tight.  Estimating $I(X; Y)$ to
$\pm \eps$ really does require $1/\eps^2$ samples (for example, if $X$
is uniform on $\{0, 1\}$ and $\Pr[Y = X] = p \approx \frac{3}{4}$,
then estimating $I(X; Y) = 1-h(p)$ requires estimating $p$ to within
$\pm \Theta(\eps)$).  We can translate this hardness into a
$\Omega(1/\eps^2)$ lower bound for a (non--tree-structured)
three-variable $P$ [see Figure~\ref{fig:hard}]; and by concatenating
$\Omega(n)$ of these instances together, we get an
$\Omega(n^2/\eps^2)$ lower bound.

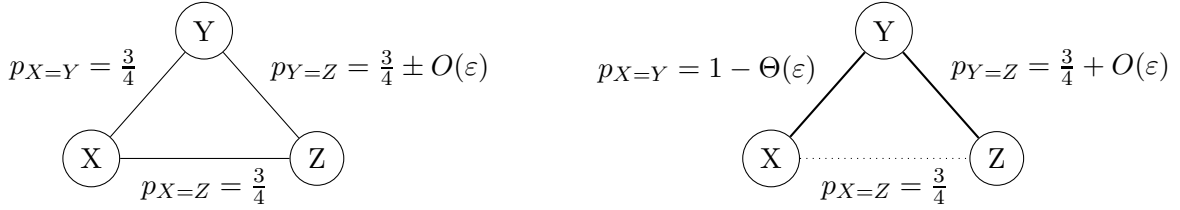
\begin{figure*}
  \centering
  \begin{subfigure}[t]{0.48\textwidth}
    \centering
  \begin{tikzpicture}
    \tikzset{vertex/.style={draw,circle, align=center}}
    \node[vertex] at (0, 0) (X) {X};
    \node[vertex] at (1.5, 1.7) (Y) {Y};
    \node[vertex] at (3, 0) (Z){Z};
    \draw (X) -- node[above left] {$p_{X=Y} = \frac{3}{4}$} ++(Y);
    \draw (Y) -- node[above right] {$p_{Y=Z} = \frac{3}{4} \pm O(\eps)$} ++(Z);
    \draw (X) -- node[below] {$p_{X=Z} = \frac{3}{4}$} ++(Z);
  \end{tikzpicture}
  \caption{Hard instance for non-realizable setting.  $X, Y$, and $Z$
    are individually uniform on $\{0, 1\}$, and pairwise match with
    probability $\approx \frac{3}{4}$.  Any $\eps$-optimal tree will
    include edge $YZ$ if $p_{Y=Z} = \frac{3}{4} + O(\eps)$ and not if
    $p_{Y=Z} = \frac{3}{4} - O(\eps)$; determining which takes
    $\Omega(1/\eps^2)$ samples.}
    \label{fig:hard}
  \end{subfigure}
  \quad
  \begin{subfigure}[t]{0.48\textwidth}
    \centering
    \begin{tikzpicture}
      \tikzset{vertex/.style={draw,circle, align=center}}
      \node[vertex] at (0, 0) (X) {X};
      \node[vertex] at (1.5, 1.7) (Y) {Y};
      \node[vertex] at (3, 0) (Z){Z};
      \draw[thick] (X) -- node[above left] {$p_{X=Y} = 1-\Theta(\eps)$} ++(Y);
      \draw[thick] (Y) -- node[above right] {$p_{Y=Z} = \frac{3}{4} + O(\eps)$} ++(Z);
      \draw[dotted] (X) -- node[below] {$p_{X=Z} = \frac{3}{4}$} ++(Z);
    \end{tikzpicture}
    \caption{A similar example in the realizable setting fails: if $P$
      is actually $X\text{-}Y\text{-}Z$-structured, and
      $p_{Y=Z}, p_{X=Z}$ are as on the left, then $p_{X=Y}$ must be
      $1-O(\eps)$.  This means that $I(\wh{Y}; \wh{Z})$ is highly
      correlated with $I(\wh{X}; \wh{Z})$, so
      $I(\wh{Y}; \wh{Z}) - I(\wh{X}; \wh{Z})$ is $\eps$-accurate with
      only $O(1/\eps)$ samples.}
    \label{fig:hardrealizable}
  \end{subfigure}
  \caption{The $\Omega(1/\eps^2)$ bound in the non-realizable setting, and its inapplicability to the realizable setting.}
\end{figure*}

\paragraph{The realizable setting.}  Fortunately, we can do much
better in the realizable setting, where $P$ is actually
$T^*$-structured for some tree $T^*$.  We show that the errors in
estimating mutual information are correlated, as illustrated in
Figure~\ref{fig:hardrealizable}, so that the \emph{difference} between
mutual informations will be estimated more accurately than the mutual
information itself.

As an example, consider the three variable case, where the true $T^*$
is $X\text{-}Y\text{-}Z$ and we want to ensure the algorithm does not pick edge $XZ$
over $YZ$.  We use the identity:
\[
  I(Y; Z) - I(X; Z) = I(Y; Z \mid X) - I(X; Z \mid Y).
\]
In order for picking $XZ$ over $YZ$ to be $\eps$-bad, the left hand
side must be at least $\eps$.  On the other hand, because $P$ is
$X\text{-}Y\text{-}Z$-structured, $I(X; Z \mid Y) = 0$, and hence
$I(Y; Z \mid X) \geq \eps$.

Chow-Liu looks at the empirical mutual information, which has the same
identity:
\begin{align}\label{eq:empiricalconditionalidentity}
  I(\wh{Y}; \wh{Z}) - I(\wh{X}; \wh{Z}) = I(\wh{Y}; \wh{Z} \mid
  \wh{X}) - I(\wh{X}; \wh{Z} \mid \wh{Y}).
\end{align}
In order for Chow-Liu to return the wrong tree by picking XZ over YZ,
this must be negative.  For this to happen, either
$I(\wh{X}; \wh{Z} \mid \wh{Y}) > \eps/10$ or
$I(\wh{Y}; \wh{Z} \mid \wh{X}) \leq \eps/10$.  This is, effectively, a
question about conditional independence testing---after how many
samples can we distinguish the conditionally independent distribution
$(X, Z \mid Y)$ from the $\eps$-far from conditionally independent
distribution $(Y, Z \mid X)$?  Our
Theorem~\ref{thm:conditionalindependencetesting} (discussed in the next
section) shows that $\Ot(1/\eps)$ samples suffice for the empirical
conditional mutual information to distinguish these cases, so
that~\eqref{eq:empiricalconditionalidentity} will be positive.

For the general $n$-variable case, consider the tree $\wh{T}$ returned by
Chow-Liu.  We can pair up the edges in $\wh{T} \setminus T^*$ with those
from $T^* \setminus \wh{T}$, such that each edge $WZ$ in $\wh{T} \setminus T^*$
is matched to an edge $XY \in T^* \setminus \wh{T}$ along the $W \leadsto Z$
path in $T^*$.  We then use the more complicated identity
\begin{align*}
  &I(X; Y) - I(W; Z)\\ &= I(X; Y) - I(X; Z) + I(X; Z) - I(W; Z) \\
                    &= I(X; Y \mid Z) - I(X; Z \mid Y) + I(Z; X \mid W) - I(Z; W \mid X).
\end{align*}
As in the three-variable case, the negative terms in this RHS are
zero, and (if picking $WZ$ over $XY$ is $\eps/n$-bad) at least one of
the positive terms is at least $\eps/(2n)$.  If this is the case, then
again Theorem~\ref{thm:conditionalindependencetesting} means that the
empirical estimates of these terms, after $\Ot(n/\eps)$ samples, will
be sufficiently accurate that
$I(\wh{X}; \wh{Y}) - I(\wh{W}; \wh{Z}) > 0$, and hence Chow-Liu will
choose $XY$ over $WZ$.  As a result, with $\Ot(n/\eps)$ samples, the
tree $\wh{T}$ recovered by Chow-Liu will satisfy~\eqref{eq:goal},
giving Theorem~\ref{thm:mainr}.

\subsection{Conditional Independence Testing}

\paragraph{Independence Testing.}
To build up to conditional independence testing with respect to mutual
information, consider \emph{unconditional} independence testing: given
samples $(X, Y) \sim P_{XY}$, determine whether $I(X; Y)$ is $0$ or
$\geq \eps$.  We would like to show that, with
$O(\frac{1}{\eps}\log\frac{1}{\eps})$ samples, the empirical mutual
information $I(\wh{X}; \wh{Y})$ will distinguish between these two
cases.  [Note that $\Omega(\frac{1}{\eps}\log\frac{1}{\eps})$ samples
are necessary even in the binary setting: if $X = Y$ always, but
$\Pr[X = 1]$ is either $0$ or $\eps / \log(1/\eps)$, the mutual
information is either $0$ or $\Theta(\eps)$, but the first
$\Omega(\frac{1}{\eps}\log\frac{1}{\eps})$ samples will probably all
be zero in either case.]

For intuition, consider the binary setting.  Let
$p_y = \Pr[X = 1 \mid Y=y]$, and $p = \Pr[X=1] = \E_y[p_y]$, so
\[
  I(X; Y) = H(X) - H(X \mid Y) = h(p) - \E_{y \sim Y}[h(p_y)]
\]
for the binary entropy function $h$.  Now, estimating either $h(p)$ or
the $h(p_y)$ to $\pm \eps$ would require $1/\eps^2$ samples: if
$p \approx \frac{1}{4}$, we would need
$\abs{\wh{p} - p} \lesssim \eps$ to estimate the individual entropies
accurately.  But if we expand $h(p_y)$ in a Taylor expansion around
$p = \E_y[p_y]$, the \emph{constant and linear terms cancel} leaving
$I(X; Y) \approx \frac{1}{2}h''(p) E_y[(p_y - p)^2]$.  So
distinguishing $I(X; Y) = 0$ from $I(X; Y) \geq \eps$ involves
distinguishing between $\E[(p_Y - p)^2] = 0 $ and
$\E[(p_Y - p)^2] \gtrsim \eps / h''(p)$.  Up to a log factor coming
from $h''(p)$, at least if the distribution of $y$ is fairly balanced,
this means it suffices to estimate each $p_y$ to within
$\pm \sqrt{\eps}/10$, which takes $O(1/\eps)$ samples.

More formally and more generally, by expressing mutual information as
KL and removing each entry in the sum's linear dependence on
$\Delta_{xy} := P_{xy} - P_xP_y$, we can write
\begin{align}\label{eq:If-overview}
  I(X; Y) =  \kl{P_{XY}}{P_XP_Y} = \sum_{x,y \in \Sigma^2} f(\Delta_{xy}, P_xP_y)
\end{align}
for some function $f$ satisfying
\[
    f(a, b) = \Theta\left(\min\left({a^2 \over b}, \abs{a} \log\left(2+{\abs{a}\over b}\right)\right)\right).
\]

We then apply Chernoff bounds to show that every individual entry of
the sum~\eqref{eq:If-overview} concentrates: in the completeness case
(Lemma~\ref{lem:completeOne}), for any $x, y$,
\[
  f(\wh{\Delta}_{xy}, \wh{P}_x\wh{P}_y) \gtrsim f(\Delta_{xy}, P_xP_y) - \frac{\log N \log(1/\delta)}{N}
\]
with probability $1-\delta$, and in the soundness case
(Lemma~\ref{lem:soundOne})
\[
  f(\wh{\Delta}_{xy}, \wh{P}_x\wh{P}_y) \lesssim \frac{\log N \log(1/\delta)}{N}.
\]
Taking a union bound over $X$ and $Y$, and plugging this into the
sum~\eqref{eq:If-overview}, gives the desired tester: as long as
$\frac{N}{\log N} \gtrsim \Sigma^2\frac{\log(\Sigma/\delta)}{\eps} $, the
empirical mutual information will distinguish between $I(X; Y) > \eps$
and $I(X; Y) = 0$.

The proofs of Lemma~\ref{lem:completeOne} and Lemma~\ref{lem:soundOne}
are somewhat technical, but straightforward.  We give intuition for
the soundness case and constant probability.  We use the two branches
of $f$ depending on whether $P_{xy} = P_xP_y$ is large or small.  If
$P_{xy} \lesssim 1/N$, then we will typically have
$\wh{\Delta}_{xy} \leq \wh{P}_{xy} \lesssim 1/N$, so
\[
  f(\wh{\Delta}_{xy}, \wh{P}_x\wh{P}_y) \lesssim \wh{\Delta}_{xy} \log N \lesssim \frac{\log N}{N}.
\]
On the other hand, if $P_{xy} \gg 1/N$, then typically
$\wh{P}_x = \Theta(P_x)$, $\wh{P}_y = \Theta(P_y)$, and
(we show) $\abs{\wh{\Delta}_{xy}} \lesssim \sqrt{P_xP_y/N}$.  Hence
\[
  f(\wh{\Delta}_{xy}, \wh{P}_x\wh{P}_y) \lesssim \frac{\wh{\Delta}_{xy}^2}{\wh{P}_x\wh{P}_y} \lesssim \frac{1}{N}.
\]

\paragraph{Conditional independence testing.}  By definition,
\[
  I(X; Y \mid Z) = \sum_z \Pr[Z=z] I(X; Y \mid Z=z).
\]

Given $N$ samples of $(X, Y, Z)$, we expect about $N\Pr[Z=z]$ samples
from $(X; Y \mid Z=z)$.  This means our unconditional mutual
independence tester will distinguish $I(X; Y \mid Z=z) = 0$ from
$I(X; Y \mid Z=z) > \frac{\abs{\Sigma}^2}{\Omegat(N \Pr[Z=z])}$.  If the distribution passes
all these independence checks, then
\[
  I(X; Y \mid Z) \leq \sum_z \Pr[Z=z]  \frac{\abs{\Sigma}^2}{\Omegat(N \Pr[Z=z])} = \frac{\abs{\Sigma}^3}{\Omegat(N)}.
\]
Thus $N = \Ot(\abs{\Sigma}^3/\eps)$ samples suffice to test conditional
independence.  A bit more care shows that the empirical conditional
mutual information works as a test statistic, achieving
Theorem~\ref{thm:conditionalindependencetesting}.

\subsection{Distribution Learning with Known Structure}

For Theorem~\ref{thm:fixed}, it is implicit
in~\cite{DBLP:journals/tit/ChowL68,DBLP:journals/ml/Dasgupta97} that
it suffices to learn the conditional distributions in KL.  While the
empirical add-1 estimator of a discrete distribution was known to
have small \emph{expected} KL error with $\Ot(\abs{\Sigma}/\eps)$
samples~\cite{pmlr-v40-Kamath15}, to our knowledge a high-probability
bound was not known.  We use a similar analysis to our independence
tester---including the same decomposition~\eqref{eq:If-overview} of
KL---to show that the empirical add-1 estimator is accurate with high
probability (Theorem~\ref{thm:KLlearn}).  We then show that our
samples from $P$ give enough samples from each individual conditional
distribution to estimate $P$ well.

\subsection{Questions}

A natural question is whether the $\abs{\Sigma}^3$ dependence in our bounds
can be improved.  The $\abs{\Sigma}^3$ term is necessary to achieve
Theorem~\ref{thm:conditionalindependencetesting} as stated; with fewer
measurements, in the soundness case of a perfectly uniform
distribution, the empirical conditional mutual information will exceed
$\eps$.  However, it is quite possible that the empirical conditional
mutual information---though $\gg \eps$---is still smaller than in the
completeness case.  Just such behavior occurs when using the empirical
total variation statistic in testing total variation from
uniformity~\cite{diakonikolas2018sample}.

Another natural question is whether one can reduce tree structure
learning to conditional independence testing as a black box.  The
Chow-Liu algorithm only considers pairwise mutual information, and
never looks at conditional mutual information at all.  Our analysis
introduces conditional mutual information
through~\eqref{eq:empiricalconditionalidentity}, which relies on the
test statistic being the empirical mutual information.  If future work
develops better conditional independence testers based on different
test statistics, does that imply more sample-efficient (but possibly
slower) algorithms for tree structure learning?

Recovering the structure of an unknown bounded degree Bayesian network remains an outstanding open question. Recently, Brustle, Cai, and Daskalakis~\cite{BCD20} have settled the sample complexity of this problem. But finding a polynomial time algorithm remains a challenge, even if we assume a correct topological ordering of the variables.
\ignore{
In this work, we focused on recovering the structure of tree structured distributions from samples. Recovering the structure of bounded degree Bayesian networks in general, have received attention more recently. Brustle et al.~\cite{DBLP:conf/sigecom/Brustle0D20} have given an exponential time algorithm for approximately learning a Bayes net of unknown structure from samples. Gao et al.~\cite{DBLP:journals/corr/abs-2006-11970} and the chain of works referenced therein, have given polynomial-time algorithms for learning the underlying directed acyclic graph of Bayesian networks from samples under certain practical assumptions. The later line of work outputs a topological ordering of the variables which is close to that of the unknown distribution and leave the work of finding the set of parents of the nodes for subsequent downstream processing. Efficiently finding such sets of parents of every variable, when their number is bounded, assuming a correct topological order of the nodes, is an interesting problem.
}

\subsection{Organization}

The rest of the paper is organized as follows.
Section~\ref{sec:prelims} describes the background and fixes notation.
Section~\ref{sec:testCI} analyzes conditional independence testing via
empirical mutual information (Theorem~\ref{thm:conditionalindependencetesting}).
Section~\ref{sec:tree} uses this to show that Chow-Liu recovers an
$\eps$-approximate tree (Theorem~\ref{thm:mainr}).  Section~\ref{sec:distribution}
shows how to recover the distribution given the tree (Theorem~\ref{thm:fixed}).
Finally, Section~\ref{sec:lower} gives lower bounds for finding an
$\eps$-approximate tree $T$, showing that Lemma~\ref{lem:mainnr} and
Theorem~\ref{thm:mainr} are nearly optimal.


\section{Notation and Preliminaries}\label{sec:prelims}

For an undirected tree $T$, a {\em rooted orientation} of $T$ fixes a root vertex and orients the edges outwards from it. For a rooted orientation of $T$, if $i$ is a vertex in $T$, $\mathrm{pa}(i)$ denotes its parent node if any, and $\text{nd}(i)$ denotes the subset of vertices not reachable from $i$.

\begin{definition}[Tree-structured distributions]\label{def:tsd}
Let $T$ be a tree. Fix any rooted orientation of it. Label the nodes of $T$ in topological order (so, node $1$ is the root). A probability distribution $P$ over $X = (X_1,\dots,X_n) \in \Sigma^n$ is said to be {\em $T$-structured} if: every variable $X_i$ is conditionally independent of $\{X_j: j \in \mathrm{nd}(i)\}$ given $X_{\mathrm{pa}(i)}$.Equivalently, $P$ admits the following factorization:
\[
\Pr[X = x] \coloneqq \Pr[X_1 = x_1] \cdot \prod_{i=2}^n \Pr[X_i = x_i \mid X_{\mathrm{pa}(i)} = x_{\mathrm{pa}(i)}] 
\]
A {\em tree-structured distribution} is $T$-structured for some tree $T$.
\end{definition}

\ignore{
We are interested in learning the structure of  a multidimensional distribution over an alphabet $\Sigma$ whose underlying dependency graph is a directed tree. We first state a known fact that in such distributions learning only the undirected graph (which we call the skeleton) of the underlying graph suffices. 

Three vertices $U,V,W$ of a directed graph $G$ is said to form a V-structure if $G$ has the edges $U \rightarrow V$ and $W \rightarrow V$ but the pair $U,W$ is not adjacent. 

\begin{theorem}
Let $G$ and $H$ be two graphs on the same vertex set so that they have the same skeleton and same V-structure. Then any probability distribution that can be represented by $G$ can also be represented by $H$ and vice versa.
\end{theorem}
}
The following classical result justifies why the rooted orientation does not matter in Definition~\ref{def:tsd}. 
\begin{theorem}[\cite{DBLP:conf/uai/VermaP90}]\label{thm:skeleton}\sutanu{this theorem is being referred}
Let $T$ be a tree on $n$ variables, and suppose $P$ is a $T$-structured distribution on $(X_1, \dots, X_n)$. For any 3 nodes $i,j,k \in [n]$, if the unique path between $i$ and $k$ in $T$ passes through $j$, then $X_i$ and $X_k$ are independent conditioned on $X_j$.
\end{theorem}

To compare distributions, we use the well-known notion of
KL-divergence.  Given two discrete probability distributions $P$ and
$Q$ over $\Sigma$, their $\dkl$-divergence is defined as\footnote{All
  logarithms in this paper are natural, so we measure
  information-theoretic quantities in nats not bits.}
\[\kl{P}{Q} \coloneqq \sum_{x \in \Sigma} P(x)\log {P(x)\over Q(x)}.\]
Recall that we say a tree $T$ is {\em $\eps$-approximate for a
  distribution $P$} if there exists a $T$-structured distribution $Q_T$
such that:
\[
\kl{P}{Q_T} \leq \eps + \min_{\text{tree } T'} \min_{\substack{T'\text{-structured}\\ \text{distribution }Q'}} \kl{P}{Q'}.
\]



\begin{algorithm}
\SetKwInOut{Input}{input}
\SetKwInOut{Output}{output}
\SetAlgoLined
\caption{Learning the Skeleton of Tree-structured Distributions from Samples}
\label{algo:algoChowLiu}
\Input{Sample access to $P$ over $X_1,\dots,X_n\in \Sigma$}
\Output{A tree $T$}
\BlankLine
$\hat{P}\gets$ the empirical distribution of $m$ i.i.d.~samples from $P$\;
\For{every $1 \le i<j\le n$}{
$I(\wh{X}_i,\wh{X}_j)\gets$ the mutual information between the variables $X$ and $Y$ with respect to $\hat{P}$\;
}
$G \gets$ the weighted complete undirected graph on $[n]$ whose edge-weight $(i,j)$ is $I(\wh{X}_i,\wh{X}_j)$\;
$T \gets$ a maximum weight spanning tree of $G$\;
\KwRet{$T$}
\end{algorithm}


The following lemma is implicit in~\cite{DBLP:journals/tit/ChowL68,DBLP:journals/ml/Dasgupta97}.  

\begin{lemma}\label{lem:ApproxChowLiu}
  For a fixed tree $T$, let $\text{pa}(v)$ denote the parent of $v$ in
  $T$ (or $\perp$ if $v$ is the root).  Let $X \sim P$ and $X' \sim Q$
  for some $T$-structured $Q$.  Then, if $\kl{P}{Q}$ is bounded:
  \begin{align*}
    &\kl{P}{Q} = \left(-H(X) + \sum_{v \in V} H(X_v)\right) - \sum_{v \in V} I(X_v; X_{\text{pa}(v)})\\&+ \sum_{v \in V} \sum_{x \in \Sigma} \Pr[X_{\text{pa}(v)} = x] \kl{X_v \mid X_{\text{pa}(v)} = x}{X'_v \mid X'_{\text{pa}(v)} = x}
  \end{align*}
\end{lemma}

Lemma~\ref{lem:ApproxChowLiu} is a generalization
of~\eqref{eq:ChowLiuEquality}: since KL divergence is nonnegative, the
$Q = P_T$ minimizing $\kl{P}{Q}$ has every entry in the final sum is
zero, which happens if $Q$ matches the marginals of $P$ on each edge
of $T$.  In that case, $\kl{P}{P_T}$ drops the final sum and
gives~\eqref{eq:ChowLiuEquality}, which we write formally:

\begin{corollary}[\cite{DBLP:journals/tit/ChowL68}]\label{cor:CLweights}
Let $P$ be a distribution over $\Sigma^n$ and $T$ be an undirected tree over the vertex set $[n]$. Let $P_T$ be the most likely distribution of $P$ for the tree $T$. Then, 
\begin{align}
  \kl{P}{P_T} = J_P - \emph{\text{wt}}_P(T)
\end{align}
where $J_P = \sum_v H(P_v) - H(P)$ is independent of $T$ ($P_v$ is the marginal on variable $v$), and
$\emph{\text{wt}}_P(T) := \sum_{(X, Y) \in T} I(X; Y)$.
\end{corollary}

This suggests the \textsf{Chow-Liu} algorithm (see
Algorithm~\ref{algo:algoChowLiu}) that we analyze.

\section{Testing Independence and Conditional Independence}\label{sec:testCI}
\paragraph{Setup.}  We assume all random variables are over a discrete
domain $\Sigma$. Let $X$ and $Y$ be random variables over $\Sigma$ distributed jointly according to $P$. For any pair $x,y \in \Sigma$, let $P_x, P_y, P_{xy}$ denote $\Pr[X=x], \Pr[Y=y], \Pr[(X,Y)=(x,y)]$,
respectively.  Let $\Delta_{xy} := P_{xy} - P_xP_y$. Hence
$\sum_{xy} \Delta_{xy} = 0$. Let $(\wh{X},\wh{Y})$ be the random variable distributed according to the \emph{empirical} distribution $\wh{P}$ over $(X,Y)$ over a finite set of independent samples. Let $\wh{P}_x, \wh{P}_y, \wh{P}_{xy}, \wh{\Delta}_{xy}$ denote the same values for $(\wh{X},\wh{Y})$.  

Define
\[
  f(a, b) := (a + b) \log (1 + a/b) - a
\]
for all $b \in [0, 1], a \in [-b, 1-b]$ [with $f(-b, b) = b$, being
the limiting value].

\begin{claim}
For two random variables $X$ and $Y$ over $\Sigma$, $I(X;Y) = \sum_{x,y} f(\Delta_{xy}, P_xP_y)$. 
\end{claim}
\begin{proof}
\begin{align}
  I(X; Y) &= D(P_{XY} || P_XP_Y)\notag\\
          &= \sum_{x,y} (P_xP_y + \Delta_{xy}) \log (1 + \Delta_{xy} / (P_xP_y))\notag\\
          &= \sum_{x,y} \left[(P_xP_y + \Delta_{xy}) \log (1 + \Delta_{xy} / (P_xP_y)) - \Delta_{xy}\right]\notag\\
          &= \sum_{x,y} f(\Delta_{xy}, P_xP_y).\label{eq:If}
\end{align}
\end{proof}

\subsection{Analysis of $f$}
\begin{lemma} \label{lem:fapprox}
For any $a\ge -b$ and $ b\ge 0$, 
  \[
    f(a, b) = C_{a,b}\min\left({a^2 \over b}, \abs{a} \log\left(2+{\abs{a}\over b}\right)\right)
  \] 
  where the coefficient $1/3\le C_{a,b} \le 1$. 
\end{lemma}
\begin{proof}
Using Calculus and Taylor expansion.
\ignore{
Let $h(z)=(1+z)\log (1+z)-z$ and $g(z)=\min(z^2,\abs{z}\log (2+\abs{z}))$. Then, $g(z)=z^2$ for $-1\le z \le 1.314$ and $g(z)=z\log (2+z)$ otherwise. We can show that 
\begin{align*}
z^2/3 \le &h(z) \le z^2 &&\tag{for $-1\le z \le 1.5$}\\
z\log (2+z)/3 \le &h(z) \le z\log (2+z)&&\tag{for $z\ge 1$} 
\end{align*}
using calculus and Taylor expansion. Hence $g(z)/3 \le h(z) \le g(z)$. We get the result upon choosing $z=a/b$.
  \eric{``verified by plotting'' ?}
  }
\end{proof}
\begin{claim}\label{claim:f-props}
For any $x,y$,  the following holds:
\begin{enumerate}
\item  $f(\Delta_{xy}, P_xP_y) \leq 1$.
\item $\min(P_x, P_y, \abs{\Delta_{xy}})\gtrsim f(\Delta_{xy}, P_xP_y) /  \log(3/f(\Delta_{xy}, P_xP_y)).$
\end{enumerate} 

\ignore{ 
  \[
    f(\Delta_{xy}, P_xP_y) \leq 1.
  \]
  and
  \[
    \min(P_x, P_y, \abs{\Delta_{xy}})\gtrsim f(\Delta_{xy}, P_xP_y) /  \log(3/f(\Delta_{xy}, P_xP_y)).
  \]
}  
\end{claim}
\begin{proof}
\noindent{\em Of (1):}  WLOG $P_x \leq P_y$.  Note that
  $-P_xP_y \leq \Delta_{xy} \leq P_x - P_xP_y$, and $f(x, a)$ is
  convex in $x$, so that:
  \begin{align*}
    f(\Delta_{xy}, P_xP_y) &\leq \max(f(-P_xP_y, P_xP_y), f(P_x-P_xP_y, P_xP_y))\\
                           &\leq \max(P_xP_y, P_x \log(1+1/P_y) - P_x + P_xP_y)\\
                           &\leq \max(1, P_x \log(1+1/P_x))\\
                           &\leq 1.
  \end{align*}
  
\noindent{\em Of (2):}  We have
  \begin{align*}
    &f(\Delta_{xy}, P_xP_y)\\
	&\lesssim \abs{\Delta_{xy}}\log\left(2+\frac{\abs{\Delta_{xy}}}{P_xP_y}\right) \leq P_x \log \left(2+\frac{1}{P_y}\right) \leq P_x \log\left(\frac{3}{P_x}\right) 
  \end{align*}
  and hence
  \[
    \min(P_x, P_y) \gtrsim f(\Delta_{xy}, P_xP_y) /  \log(3/f(\Delta_{xy}, P_xP_y)).
  \]
  But then
  \begin{align*}
    \abs{\Delta_{xy}} &\gtrsim  f(\Delta_{xy}, P_xP_y) / \log\left({\frac{3}{P_y}}\right)\\ 
	&\gtrsim f(\Delta_{xy}, P_xP_y) / \log (3/f(\Delta_{xy}, P_xP_y)),
  \end{align*}
  finishing the result.
\end{proof}

\subsection{Properties of the Empirical Distribution}
By Chernoff bounds, the empirical distribution is close to the actual one:
\begin{claim}\label{claim:Pempirical}
  Let $\wh{P}_x, \wh{P}_y$, and $\wh{P}_{xy}$ be empirical
  distributions over $N > 1$ samples. Then with $1-3\delta$
  probability, all of the following bounds hold:
  \begin{align*}
    \abs{\wh{P}_x - P_x} &\lesssim \sqrt{ {P_x\log {2\over \delta}\over N}}+{\log {2\over \delta} \over N}\\
    \abs{\wh{P}_y - P_y} &\lesssim \sqrt{ {P_x\log {2\over \delta}\over N}}+{\log {2\over \delta} \over N}\\
    \abs{\wh{P}_{xy} - P_{xy}} &\lesssim \sqrt{ {P_{xy}\log {2\over \delta}\over N}}+{\log {2\over \delta} \over N}\\
    \abs{\wh{P}_x\wh{P}_y - P_xP_y} &\lesssim \sqrt{P_xP_y{\log {2\over \delta}\over N}}+(P_x+P_y){\log {2\over \delta} \over N}+{\log^2 {2\over \delta} \over N^2}\\
    \abs{\wh{\Delta}_{xy} - \Delta_{xy}} &\lesssim \sqrt{\abs{\Delta_{xy}}\frac{\log {2\over\delta}}{N}} + \sqrt{P_{x}P_{y}\frac{\log {2\over\delta}}{N}}+{\log {2\over \delta} \over N}+{\log^2 {2\over \delta} \over N^2}.
  \end{align*}
\end{claim}
\begin{proof}
\ignore{
  We have that
  \begin{align*}
    \E[\wh{P}_x] &= P_x\\
    \Var[\wh{P}_x] &= P_x(1-P_x)/N \leq P_x/N.
  \end{align*}
  Chebyshev's inequality then gives the first equation with $1-\delta$
  probability.  The same holds for the second and third equations, so
  with $1-3\delta$ probability they all hold.  
}

By the multiplicative Chernoff bound,
\[
  \Pr[|\widehat{P}_x - P_x| > \eps P_x]< 2\exp(-C\min(\eps, \eps^2) P_x N).
\]
Rearranging, with probability $1-\delta$,
\begin{align*}
|\wh{P}_x - P_x| &\lesssim \max\left({\log {2\over \delta} \over P_xN}, \sqrt{{\log {2\over \delta} \over P_xN}}\right) P_x
	      \le \sqrt{ {P_x\log {2\over \delta}\over N}}+{\log {2\over \delta} \over N}.
\end{align*}

Similarly for $\wh{P}_y$ and $\wh{P}_{xy}$. Then
  \begin{align*}
    &\abs{\wh{P}_x\wh{P}_y - P_xP_y}\\ &\lesssim {\sqrt{\log {2\over \delta}\over N}}(P_x\sqrt{P_y} + P_y\sqrt{P_x}) + \sqrt{P_xP_y}{\log {2\over\delta}\over N}+(P_x+P_y){\log {2\over \delta} \over N}+\\&\qquad\qquad\qquad\qquad\qquad\qquad\qquad(\sqrt{P_x}+\sqrt{P_y}){\log^{1.5} {2\over \delta} \over N^{1.5}}+{\log^2 {2\over \delta} \over N^2}\\
&\le  2\sqrt{P_xP_y{\log {2\over \delta}\over N}}+2(P_x+P_y){\log {2\over \delta} \over N}+2(\sqrt{P_x + P_y}){\log^{1.5} {2\over \delta} \over N^{1.5}}+\\
&\qquad\qquad\qquad\qquad\qquad{\log^2 {2\over \delta} \over N^2}\\
&\le  2\sqrt{P_xP_y{\log {2\over \delta}\over N}}+3(P_x+P_y){\log {2\over \delta} \over N}+2{\log^2 {2\over \delta} \over N^2}.
  \end{align*}
  This implies:
  \begin{align*}
    \abs{\wh{\Delta}_{xy} - \Delta_{xy}} &\leq \abs{\wh{P}_{xy} - P_{xy}} + \abs{\wh{P}_x\wh{P}_y - P_xP_y}\\
    &\lesssim \sqrt{P_{xy}\frac{\log {2\over\delta}}{N}} + \sqrt{P_{x}P_{y}\frac{\log {2\over\delta}}{N}}+{\log {2\over \delta} \over N}+{\log^2 {2\over \delta} \over N^2}.
  \end{align*}
  The result follows because $P_{xy} \leq 2 \Delta_{xy}$  whenever that term is largest.
\end{proof}

\subsection{Completeness}
\begin{lemma}\label{lem:completeOne}
  Let $\wh{P}$ be the empirical distribution over $N$ samples.  Then there exist constants $C, C' > 0$ such that for every $\delta > 0$  if $f(\Delta_{xy}, P_xP_y) \geq C \frac{\log N}{N}\log {2\over \delta}$, then:
  \[
    \Pr[f(\wh{\Delta}_{xy}, \wh{P}_x\wh{P}_y) > C' f(\Delta_{xy},
    P_xP_y)] > 1-3\delta.
  \]
\end{lemma}
\begin{proof}

  By Claim~\ref{claim:f-props} $f(\Delta_{xy}, P_xP_y) \leq 1$.  Claim~\ref{claim:f-props} also implies
  \begin{align}
    \label{eq:Deltabig}
    \min(\abs{\Delta_{xy}}, P_x, P_y) \gtrsim \frac{f(\Delta_{xy}, P_xP_y)}{\log (2/f(\Delta_{xy}, P_xP_y))} \gtrsim \frac{C}{N}\log {2\over \delta}.
  \end{align}


  Suppose that the Claim~\ref{claim:Pempirical} statements hold, as
  happens with $1-3\delta$ probability.  We will show that this
  implies the result.

  We split into cases, based on whether $\Delta_{xy} > 8P_xP_y$.

  \paragraph{Large $\Delta_{xy}$.}  This case of $\Delta_{xy} > 8 P_x P_y$ implies
  \[
    f(\Delta_{xy}, P_xP_y) \eqsim \Delta_{xy}\log \frac{\Delta_{xy}}{P_xP_y}.
  \]

  In this regime, we have by Claim~\ref{claim:Pempirical} holding that
  \[
    \abs{\wh{\Delta}_{xy} - \Delta_{xy}}  \lesssim \sqrt{\frac{\Delta_{xy}}{N}\log {2\over \delta}}+{\log {2\over \delta}\over N}.
  \]
  Since $N \gtrsim C\log {2\over \delta}/\Delta_{xy}$ by~\eqref{eq:Deltabig}, this implies
  \[
    \abs{\wh{\Delta}_{xy} - \Delta_{xy}} \lesssim \Delta_{xy}/C
  \]
  and hence $\abs{\wh{\Delta}_{xy} - \Delta_{xy}} < \Delta_{xy}/10$
  for a sufficiently large $C$.

  We also have by Claim~\ref{claim:Pempirical} holding that
  \begin{align*}
    \abs{\wh{P}_x\wh{P}_y - P_xP_y} &\lesssim \sqrt{P_xP_y{\log {2\over \delta}\over N}}+(P_x+P_y){\log {2\over \delta} \over N}+{\log^2 {2\over \delta} \over N^2} \\&\lesssim \sqrt{\frac{\Delta_{xy}}{N}\log {2\over \delta}}+{\log {2\over \delta}\over N}
  \end{align*}
and hence (by~\eqref{eq:Deltabig}) $\abs{\wh{P}_x\wh{P}_y - P_xP_y} \le  \Delta_{xy}/10$ for a sufficiently large $C$. This implies $\wh{P}_x\wh{P}_y \leq 0.23 \Delta_{xy}$.  Therefore:
  \[
    \frac{\wh{\Delta}_{xy}}{\wh{P}_x\wh{P}_y} \geq \frac{0.9
      \Delta_{xy}}{0.23 \Delta_{xy}} > 3.9,
  \]
  so that (in Lemma~\ref{lem:fapprox}),
  \begin{align*}
    f(\wh{\Delta}_{xy}, \wh{P}_x\wh{P}_y)
    &\eqsim \wh{\Delta}_{xy}\log \frac{\wh{\Delta}_{xy}}{\wh{P}_x\wh{P}_y}
    \gtrsim \Delta_{xy}\log \frac{\Delta_{xy}}{\wh{P}_x\wh{P}_y}
  \end{align*}
Now,
\begin{align*}
  \abs{\wh{P}_x\wh{P}_y - P_xP_y} &\lesssim \sqrt{P_xP_y{\log {2\over \delta}\over N}}+(P_x + P_y){\log {2\over \delta} \over N}+{\log^2 {2\over \delta} \over N^2}\\
                                  &\lesssim \sqrt{P_xP_y \Delta_{xy}/C}+2P_xP_y/C + P_xP_y/C^2 &&\tag{Using \eqref{eq:Deltabig}}\\
                                  &\lesssim \frac{1}{\sqrt{C}}(P_xP_y + \sqrt{P_xP_y\Delta_{xy}}).
\end{align*}
For sufficiently large constant $C$ the constant factor is overcome, so that
\begin{align*}
  f(\wh{\Delta}_{xy}, \wh{P}_x\wh{P}_y) &\gtrsim \Delta_{xy} \log \frac{\Delta_{xy}}{2P_xP_y+\sqrt{P_xP_y \Delta_{xy}}}\\
                                        &\geq \Delta_{xy} \min(\log \frac{\Delta_{xy}}{4P_xP_y},
                                          \log \frac{\Delta_{xy}}{2\sqrt{P_xP_y \Delta_{xy}}})\\
                                        &\eqsim \Delta_{xy} \min(\log \frac{\Delta_{xy}}{P_xP_y},
                                          \frac{1}{2}\log \frac{\Delta_{xy}}{P_xP_y})\\
  &\eqsim f(\Delta_{xy}, P_xP_y)
  \end{align*}
  as desired.
  \paragraph{Small $\Delta_{xy}$.}  This case of $- P_x P_y \le \Delta_{xy} \leq 8 P_x P_y$ implies
  \[
    f(\Delta_{xy}, P_xP_y) \eqsim \Delta_{xy}^2/(P_xP_y) \leq 64 P_xP_y.
  \]

  Now, the condition that
  $f(\Delta_{xy}, P_xP_y) > C \frac{\log N}{N}\log {2\over \delta}$ implies
  \begin{align}\label{eq:smalldeltacondition}
    P_xP_y \ge \frac{1}{64}\Delta_{xy}^2/(P_xP_y) \gtrsim C \frac{\log N}{N}\log {2\over \delta}
  \end{align}
  and hence $N \gtrsim C\log {2\over \delta}/\min(P_x, P_y)$.  Therefore, for
  a sufficiently large $C$, we have by Claim~\ref{claim:Pempirical} that both:
  \begin{align*}
    \abs{\wh{P}_x - P_x} &\leq P_x/10\\
    \abs{\wh{P}_y - P_x} &\leq P_y/10.
  \end{align*}
  Furthermore, the condition~\eqref{eq:smalldeltacondition} also implies:
  \begin{align}\label{eq:smalldeltacondition2}
    \abs{\Delta_{xy}} \gtrsim \sqrt{\frac{C P_xP_y }{N}\log {2\over \delta}}.
  \end{align}
  Hence by Claim~\ref{claim:Pempirical} and the conditions~\eqref{eq:smalldeltacondition} and~\eqref{eq:smalldeltacondition2},
  \begin{align*}
    \abs{\wh{\Delta}_{xy} - \Delta_{xy}} &\lesssim \sqrt{\frac{P_xP_y}{N}\log {2\over \delta}} + {\log {2\over \delta}\over N}
    \lesssim \abs{\Delta_{xy}}/\sqrt{C} + \Delta_{xy}^2/(P_xP_yC).
  \end{align*}
  Using $\abs{\Delta_{xy}}\le 8P_xP_y$, we get $\abs{\wh{\Delta}_{xy} - \Delta_{xy}} < \abs{\Delta_{xy}}/10$ for a sufficiently large constant $C$.  Therefore:
  \[
  f(\wh{\Delta}_{xy}, \wh{P}_x\wh{P}_y) \gtrsim \frac{\wh{\Delta}_{xy}^2}{\wh{P}_x\wh{P}_y} \eqsim \frac{{\Delta}_{xy}^2}{{P}_x{P}_y} \eqsim f(\Delta_{xy}, P_xP_y).
\]
\end{proof}

\begin{corollary}\label{cor:completeness}
  Let $\wh{P}$ be the empirical distribution over $N > 1$ samples.  Then there exist universal constants $C_1, C_2 > 0$ such that
  for every $\delta > 0$:
  \[
    I(\wh{X};\wh{Y}) \ge C_1 I(X;Y) - C_2|\Sigma|^2{\log N\over N}\log {\abs{\Sigma}
        \over \delta}
  \]
  with probability at least $1-\delta$.
\end{corollary}
\begin{proof}
  Lemma~\ref{lem:completeOne} has a condition on $f$ being large.  But
  in general, since $f \geq 0$ always, it shows that with probability $1-3\delta$,
  \[
    f(\wh{\Delta}_{xy}, \wh{P}_x\wh{P}_y) > C' f(\Delta_{xy},
    P_xP_y) - C'C \frac{\log N}{N} \log \frac{2}{\delta}.
  \]
  Taking a union bound over the sum~\eqref{eq:If}, and rescaling
  $\delta$ by $3 \abs{\Sigma}^2$, we get the result.
\end{proof}

\subsection{Soundness}
\begin{lemma}\label{lem:soundOne}
  Let $\wh{P}$ be the empirical distribution over $N$ samples.  Then
  there exists a universal constant $C > 0$
  such that for every $\delta > 0$, if $\Delta_{xy} = 0$ then:
  \[
    \Pr[f(\wh{\Delta}_{xy}, \wh{P}_x\wh{P}_y) < C \frac{\log N}{N}{\log {2\over \delta}}] > 1-3\delta.
  \]
\end{lemma}
\begin{proof}
  As in the completeness section, we suppose that the equations of
  Claim~\ref{claim:Pempirical} all hold.  In particular, this implies
  that
  \[
    \abs{\wh{\Delta}_{xy}} \lesssim \sqrt{\frac{P_xP_y}{N}{\log {2\over \delta}}}+{\log {2\over \delta}\over N}.
  \]

  We again split into cases depending on
  $P_xP_y < \frac{\log {2\over \delta}}{CN}$ or not for a large enough
  constant $C$.

  \paragraph{Small $P_xP_y$.}  Suppose
  $P_xP_y < \frac{\log {2\over \delta}}{CN}$, so that
  \[
    \abs{\wh{\Delta}_{xy}} \lesssim \frac{1}{N}{\log {2\over \delta}}.
  \]

  First note that, if either of $\wh{P}_x = 0$ or $\wh{P}_y = 0$, then
  $\wh{P}_{xy} = 0 $ and $\wh{\Delta}_{xy} = 0$, so
  $f(\wh{\Delta}_{xy}, \wh{P}_x\wh{P}_y) = 0$.  Therefore, in order
  for $f(\wh{\Delta}_{xy}, \wh{P}_x\wh{P}_y)$ to be nonzero, we must
  sample $x$ and $y$ in our set, in which case
  $\wh{P}_x\wh{P}_y \geq 1/N^2$.

  Therefore
  \begin{align*}
    f(\wh{\Delta}_{xy}, \wh{P}_x\wh{P}_y) &\lesssim
    \abs{\wh{\Delta}_{xy}} \log (1 +
    \wh{\Delta}_{xy}/(\wh{P}_x\wh{P}_y))\\ &\lesssim
    \abs{\wh{\Delta}_{xy}} \log (1 + N^2) \lesssim \frac{1}{N}{\log {2\over \delta}}\log N.
  \end{align*}
  \paragraph{Large $P_xP_y$.}  If $P_xP_y \ge \frac{\log {2\over \delta}}{CN}$, then
  \[
    \min(P_x, P_y) \geq \frac{\log {2\over \delta}}{CN}
  \]
  and hence by Claim~\ref{claim:Pempirical} holding we have
  \begin{align*}
    \abs{\wh{P}_x - P_x} \leq O\left(\sqrt{\frac{P_x}{N}\log {2\over \delta}} + {\log {2\over \delta} \over N} \right)\leq P_x/2,
  \end{align*}
  for a large enough $C$ and similarly for $\wh{P_y}$.  Therefore:
  \[
    f(\wh{\Delta}_{xy}, \wh{P}_x\wh{P}_y) \lesssim \frac{\wh{\Delta}_{xy}^2}{\wh{P}_x\wh{P}_y} \leq 4\frac{\wh{\Delta}_{xy}^2}{P_xP_y} \lesssim \frac{1}{N}\log {2\over \delta}+ {\log^2 {2\over \delta}\over N^2 P_x P_y} \lesssim \frac{C}{N}\log {2\over \delta}.
  \]

  Therefore the result holds regardless of the case, as long as
  Claim~\ref{claim:Pempirical} holds.
\end{proof}

\begin{corollary}\label{cor:soundness}
  Let $\wh{P}$ be the empirical distribution over $N > 1$ samples.  There
  exists a universal constant $C_3 > 0$ such that if
  $P$ is a product distribution, then for every $\delta > 0$:
  \[
   I(\wh{X};\wh{Y}) \le \frac{\log N}{N} C_3|\Sigma|^2\log {\abs{\Sigma} \over \delta}
  \]
  with probability at least $1-\delta$.
\end{corollary}
\begin{proof}
  Follows from taking the sum~\eqref{eq:If} and applying a union bound
  over the events in Lemma~\ref{lem:soundOne} for all possible $x,y$.
\end{proof}

\ignore{
\emph{Sketch of the rest of the soundness case.} To then show that
$I(\wh{X}; \wh{Y})$ is small, we can set $\delta' = \delta/\Sigma^2$
and get that with probability $1-\delta$ we have
\[
  I(\wh{X}; \wh{Y}) \leq \sum f(\wh{\Delta}_{xy}, \wh{P}_x\wh{P}_y) \lesssim \Sigma^4 \log N / N.
\]

The dependence on $\Sigma$ here is really bad.  You should certainly
be able to improve the lemma to have a $\log(1/\delta)$ dependence.
But ideally you could get an expectation and avoid the union bound
entirely.
}
\sutanu{we should emphasize the non-conditional independence tester also}
\subsection{Conditional Independence Testing}

\condintest*
\begin{proof}
For any $z\in \Sigma$ let $N_z$ be the number of samples with $Z=z$. 

\medskip
\noindent{\em Proof of (1)}: If $I(X;Y \mid Z) = 0$ then
$I(X; Y \mid Z=z) = 0$ for each $z$. Then
Corollary~\ref{cor:soundness} gives us that, with probability at
least $1-\delta$,
\[
  I(\wh{X}; \wh{Y} \mid \wh{Z}=z) \lesssim \frac{|\Sigma|^2}{N_z}{\log
    {\abs{\Sigma}\over \delta}}\log N_z \le \frac{|\Sigma|^2}{N_z}{\log
    {\abs{\Sigma}\over \delta}}\log N.
\]
Let $S \subseteq \Sigma$ contain the set of $z$ such that
$\Pr[\abs{\wh{P}_z - P_z} > P_z/2] \leq \delta$.  By a Chernoff bound, this consists
of all $z$ with $P_z \geq O(\frac{\log 1/\delta}{N})$.  With probability
$1-2\abs{\Sigma}\delta$, then,
\begin{align*}
  \sum_{z \in S} \wh{P}_z I(\wh{X}; \wh{Y} \mid \wh{Z}=z) &\lesssim \sum_{z \in S} P_z \frac{|\Sigma|^2}{N P_z/2}{\log
    {\abs{\Sigma}\over \delta}}\log N\\
  &\leq \frac{\log N}{N} 2|\Sigma|^3\log
    {\abs{\Sigma}\over \delta}.
\end{align*}
On the other hand, $z \notin S$ will have
$\wh{P}_z \lesssim \frac{\log (1/\delta)}{N}$ with probability
$1-\delta$, so that with probability $1-\abs{\Sigma} \delta$
\[
  \sum_{z \notin S} \wh{P}_z I(\wh{X}; \wh{Y} \mid \wh{Z}=z) \lesssim \abs{\Sigma} \frac{\log (1/\delta)}{N} \log \abs{\Sigma} \lesssim \frac{\abs{\Sigma}^2 \log (1/\delta)}{N}
\]
is even smaller.  Rescaling $\delta$, we get with probability $1-\delta$ that
\[
  I(\wh{X}; \wh{Y} \mid \wh{Z}) = \sum_{z} \wh{P}_z I(\wh{X}; \wh{Y} \mid \wh{Z}=z) \lesssim \frac{\log N}{N} |\Sigma|^3\log
    {\abs{\Sigma}\over \delta}
\]
which is at most $\eps$ for the desired $N$.

\medskip
\noindent{\em Proof of (2)}:
Consider the set $S$ of $z\in \Sigma$ which satisfy
$P_{z} \times I(X; Y \mid Z=z^*) \geq \frac{I(X;Y\mid Z)}{2\abs{\Sigma}}$. Note that this implies $P_{z} \ge \frac{\eps}{2\abs{\Sigma}\log \abs{\Sigma}}$. We also have,
\begin{align*}
&\sum_{z\in S} P_z\times I(X;Y \mid Z=z)\\
&=\sum_{z} P_z\times I(X;Y \mid Z=z) - \sum_{z\notin S} P_z\times I(X;Y \mid Z=z)\\
&\ge I(X;Y \mid Z) - |\Sigma|\frac{I(X,Y\mid Z)}{2|\Sigma|}\\
&\ge I(X,Y\mid Z)/2.
\end{align*}

Our $N$ is large enough that for $z \in S$,
$\E[N_z] = P_z N \gtrsim \log (\abs{\Sigma}/\delta)$.  Hence, with probability
$1-\delta$, we have $N_z \geq N P_z/2$ for all $z \in S$.
Then Corollary~\ref{cor:completeness} gives us, with probability $1-\delta$,
\[
I(\wh{X}; \wh{Y} \mid \wh{Z}=z) \geq C_1I(X; Y \mid Z=z)-2C_2\abs{\Sigma}^2{\log ({0.5NP_{z}})\over NP_{z}} \log {\abs{\Sigma} \over \delta},
\]
for all $z \in S$. Multiplying $P_{z}$ and summing over all $z \in S$
give us:
\begin{align*}
  I&(\wh{X};\wh{Y}\mid \wh{Z}) \ge \sum_{z \in S} \wh{P}_zI(\wh{X};\wh{Y}\mid \wh{Z}=z)\\
                              &\ge \sum_{z \in S} \frac{P_z}{2}\left(C_1I(X; Y \mid Z=z)-2C_2\abs{\Sigma}^2{\log ({0.5NP_{z}})\over NP_{z}} \log {\abs{\Sigma} \over \delta}\right)\\
                              &\geq \frac{C_1}{4}I(X; Y \mid Z) - C_2 \abs{\Sigma}^3 \frac{\log N}{N} \log \frac{\abs{\Sigma}}{\delta}.
\end{align*}
For $N$ as large as given, the term being subtracted is at most
$\frac{\eps C_1}{8}$, which is at most half the first term.
\end{proof}

A (non-conditional) independence tester with a sample complexity of $|\Sigma|^{-1}$-factor of that of Theorem~\ref{thm:conditionalindependencetesting} closely follows from the above proof.


\section{Tree Structure Recovery}\label{sec:tree}
\subsection{Non-realizable Case}
Let $P$ be an unknown distribution over $\Sigma^n$ and $\hat{P}$ be the empirical distribution of $P$ for a certain number of samples to be fixed later. Our algorithm \textsf{Chow-Liu} returns a maximum spanning tree $\hat{T}$ of the complete graph whose edge weights for every pair of variables are given by the estimated mutual informations with respect to $\hat{P}$.

Let $T^*$ be the tree minimizing $\kl{P}{P_{T^*}}$.  Recall that
$\text{wt}_P(T)$ is defined as the sum of the pairwise mutual
informations across $T$.  By Corollary~\ref{cor:CLweights}, the \textsf{Chow-Liu} algorithm will
return a tree $\wh{T}$ satisfying
\[
  \kl{P}{P_{\wh{T}}} \leq \kl{P}{P_{T^*}} + \eps
\]
if $\text{wt}_P(\wh{T}) \geq \text{wt}_P(T^*) - \eps$.  Since $\wh{T}$
maximizes $\text{wt}_{\wh{P}}(\wh{T})$, it would suffice to ensure
$\text{wt}_{\wh{P}}(T) = \text{wt}_{P}(T) \pm \eps/2$ for all $T$; and
therefore it would suffice for
\[
  I(\wh{X}; \wh{Y}) = I(X; Y) \pm \frac{\eps}{2n}
\]
for all pairs of variables $(X, Y)$.  The following result is standard, which analyzes the
the plug-in estimator $H(\wh{X})$ for estimating a single discrete
entropy $H(X)$ to $\pm \eps$ with probability $1-\delta$.
\begin{theorem}[\cite{antos2001convergence,paninski2003estimation}]\label{thm:estEntropy}
For $N \gtrsim \left(\frac{\abs{\Sigma}}{\eps} \allowbreak + \frac{1}{\eps^2}\log
  \frac{1}{\delta}\log^2\allowbreak (\frac{\abs{\Sigma}}{\eps}\log
  \frac{1}{\delta})\right)$, $\abs{H(\wh{X})-H(X)}\le \eps$ with probability at least $(1-\delta)$.
\end{theorem}
Since $I(X; Y) = H(X) + H(Y) - H(X, Y)$, Theorem~\ref{thm:estEntropy}
tells us that once
\[N \gtrsim \left(\frac{\abs{\Sigma}^2n}{\eps} \allowbreak+ \frac{n^2}{\eps^2}\log
\frac{n}{\delta}\log^2\allowbreak (\frac{n \abs{\Sigma}}{\eps}\log \frac{n}{\delta})\right),\] all the
pairwise mutual informations of the variables of $P$ will be estimated to within $\frac{\eps}{2n}$. In that case, \textsf{Chow-Liu} would return a tree $T$, the best distribution  on which would be close to the closest tree Bayes net of $P$.

\begin{lemma}\label{lem:CLquadraticUnreal}
  Let $P$ be any unknown distribution over $\Sigma^n$. Let $Q$ be the
  tree Bayes net which is closest to $P$ in $\dkl$ distance.  Then
  \emph{\textsf{Chow-Liu}}, when run with
  $O\left(\frac{\abs{\Sigma}^2n}{\eps} + \frac{n^2}{\eps^2}\log
    \frac{n}{\delta}\log^2 (\frac{n\abs{\Sigma}}{\eps}\log
    \frac{n}{\delta})\right)$ samples, returns a tree $T$ such that there
  exists a $T$-structured $R$ with $\kl{P}{R} \leq \kl{P}{Q} +\eps$.
\end{lemma}

We conclude this section by noting that when $P$ itself is a tree Bayes net (realizable case) $ \kl{P}{Q}=0$ and the best Bayes net $R$ on the tree returned by $\textsf{Chow-Liu}$ with the sample complexity of Lemma~\ref{lem:CLquadraticUnreal} would satisfy $\kl{P}{R} \leq \eps$ with probability at least $(1-\delta)$. In the next section, we show how to bring the sample complexity analysis down
from $\Ot(n^2/\eps^2)$ to $\Ot(n/\eps)$ when $P$ is actually
tree-structured for some unknown tree.
\subsection{Realizable Case}
We will need the following fact which follows from the chain rule of
mutual information.
\begin{fact}
For three random variables $X, Y,$ and $Z$
\[ I(X; Y) - I(X; Z) = I(X; Y \mid Z) - I(X; Z \mid Y).\]
\end{fact}
\begin{proof}
Follows from observing $ I(X; Y)+ I(X; Z \mid Y)  =  I(X;Y, Z) \allowbreak = I(X; Z) + I(X; Y \mid Z) $.
\end{proof}

We also use the following fact about spanning trees:
\begin{fact}\label{fact:swapTree}
  Let $T_1$ and $T_2$ be two spanning trees on $n$ vertices such that
  their symmetric difference consists of the edges
  $E=\{e_1,e_2\dots,e_{\ell}\}\allowbreak\in T_1\setminus T_2$ and
  $F=\{f_1,f_2\dots,f_{\ell}\}\in T_2\setminus T_1$. Then $E$ and $F$
  can be paired up, without loss of generality say
  $\langle e_i,f_i \rangle$, such that for all $i$,
  $T_1 \cup \{f_i\} \setminus \{e_i\}$ is a spanning tree.
\end{fact}
\begin{proof}
  We use induction on $\ell$. Base case of $\ell=0$ is trivial.

  Assume it holds for any two trees $T_1$ and $T_2$ so that
  $|T_1\setminus T_2| = |T_2\setminus T_1| = (l-1)$.  Now, pick an
  arbitrary $e = (u, v) \in T_1 \setminus T_2$.  $T_1 \setminus \{e\}$
  has two connected components, $L \ni u$ and
  $R = (V \setminus L) \ni v$.  In $T_2$, there is some path
  connecting $u$ to $v$.  This path starts in $L$ and ends in $R$, so
  it must have some edge $f$ connecting $L$ to $R$.  But
  $f \notin T_1$, since $e \notin T_2$ is the only edge connecting $L$
  and $R$ in $T_1$.

  Because $f$ connects $L$ and $R$, which are otherwise unconnected in
  $T_1 \setminus \{e\}$, $T_1 \cup\{f\} \setminus \{e\}$ is a spanning
  tree.  Thus it is valid to pair $\langle e, f\rangle$.  Furthermore,
  because $f$ lies on the path connecting $e$ in $T_2$,
  $T_3 := T_2 \cup \{e\} \setminus \{f\}$ is also a spanning tree, and
  it differs in only $\ell - 1$ edges from $T_1$.  Therefore by
  induction, $T_3$ can be paired with $T_1$ in the desired way.
  Adding $\langle e, f \rangle$ to this pairing means that $T_2$ can
  be paired with $T_1$.
\end{proof}

\begin{theorem}
  Let $N$ be such that the bound in
  Theorem~\ref{thm:conditionalindependencetesting} holds for a given
  $\eps, \delta > 0$. Then in the realizable case with $n$ nodes, with
  probability $1 - 4n\delta$ $\emph{\textsf{Chow-Liu}}$ returns a tree $\wh{T}$ with
  $\kl{P}{P_{\wh{T}}} \leq \eps n$.
\end{theorem}
\begin{proof}
  Let $P$ be the unknown distribution on the true tree $T^*$.  Let
  $\wh{T}$ be the tree returned by the \textsf{Chow-Liu} algorithm with $N$
  samples.  For any set of vertex pairs $S$, let $\weight(S)$ and
  $\wh{\weight}(S)$ denote the sum of mutual information over all pairs in
  $S$ with respect to the true and empirical distributions
  respectively, so $\weight(T^*)$ and $\wh{\weight}(\wh{T})$ are each maximal
  over spanning trees.

  For our analysis, we make at most $4n$ invocations of
  Theorem~\ref{thm:conditionalindependencetesting}. We will assume the
  conclusion holds in all cases, as happens with at least $1-4n\delta$
  probability.

  Let $\{\langle e_i,f_i \rangle\}_i$ be a pairing given by
  Fact~\ref{fact:swapTree} for $T^*$ and $\wh{T}$. By~\eqref{eq:ChowLiuEquality},
  
  \[
    \kl{P}{P_{\wh{T}}}=\weight(T^*)-\weight(\wh{T})=\sum_i (\weight(\{e_i\})-\weight(\{f_i\})).
  \]

  For any $i$, let $e_i=(X_i,Y_i) \in T^*$ and
  $f_i=(W_i,Z_i) \in \wh{T}$. Because
  $T^* \cup \{f_i\} \setminus \{e_i\}$ is a spanning tree, the path
  connecting $W_i$ and $Z_i$ in $T^*$ must go through $e_i$.  Without
  loss of generality let $W \leadsto X\rightarrow Y\leadsto Z$ be this
  path (it is possible that $W=X$ or $Y=Z$). Hence, from Theorem~\ref{thm:skeleton}, we have for the true
  distribution

\[I(X_i; Z_i \mid Y_i) = I(Z_i; W_i \mid X_i) = 0\]
and so
\begin{align}\label{algn:trueEq}
 \nonumber   &I(X_i; Y_i) - I(W_i; Z_i) \\&= I(X_i;Y_i) - I(X_i;Z_i) + I(X_i;Z_i) - I(W_i;Z_i) \\
 \nonumber &=    I(X_i;Y_i\mid Z_i) - I(X_i;Z_i\mid Y_i) + I(X_i;Z_i\mid W_i) - I(Z_i;W_i\mid X_i) \\
&=     I(X_i; Y_i \mid Z_i) + I(Z_i; X_i \mid W_i).
\end{align}
On the other hand the empirical distribution will have
  \begin{align}\label{eqn:empTree}
    &I(\wh{X_i}; \wh{Y_i}) - I(\wh{W_i}; \wh{Z_i})\\
    &= I(\wh{X_i}; \wh{Y_i} \mid
      \wh{Z_i}) - I(\wh{X_i}; \wh{Z_i} \mid \wh{Y_i}) + I(\wh{Z_i}; \wh{X_i} \mid
      \wh{W_i}) - I(\wh{Z_i}; \wh{W_i} \mid \wh{X_i}).
  \end{align}

Because $\wh{T}$ is maximal under $\wh{\weight}$,
\begin{align}
  0 &\geq \wh{\weight}(T^*)-\wh{\weight}(\wh{T})\nonumber\\&=\sum_i (    I(\wh{X_i}; \wh{Y_i}) - I(\wh{W_i}; \wh{Z_i}))\notag\\
&= \sum_i \left(I(\wh{X_i}; \wh{Y_i} \mid
      \wh{Z_i}) + I(\wh{Z_i}; \wh{X_i} \mid
      \wh{W_i})\right) - \nonumber\\
	  &\qquad\qquad\qquad\sum_i \left(I(\wh{X_i}; \wh{Z_i} \mid \wh{Y_i}) + I(\wh{Z_i}; \wh{W_i} \mid \wh{X_i})\right).\label{eq:whdiff}
\end{align}

We invoke Theorem~\ref{thm:conditionalindependencetesting} with $\eps':=C \eps /10$ where $C < 1$ is the constant given by the theorem.  As a consequence of Theorem~\ref{thm:conditionalindependencetesting} and the fact that each of $I(X_i;Z_i \mid Y_i) = I(Z_i; W_i \mid X_i) = 0$, the second sum is at most $C\eps n/5$.

On the other hand, Theorem~\ref{thm:conditionalindependencetesting} implies that
\[
  I(\wh{X_i}; \wh{Y_i} \mid \wh{Z_i}) \geq C  I(X_i; Y_i \mid Z_i)  - C \eps'
\]
and similarly for $I(\wh{Z_i}; \wh{X_i} \mid \wh{W_i})$.  As a result, the first sum has
\begin{align*}
  \sum_i &\left(I(\wh{X_i}; \wh{Y_i} \mid \wh{Z_i}) + I(\wh{Z_i};
  \wh{X_i} \mid \wh{W_i})\right)\\ &\geq C \sum_i (I(X_i; Y_i \mid Z_i) + I(Z_i; X_i \mid W_i) - 2 \eps')\\
                                 &\geq C (\weight(T^*)-\weight(\wh{T})) - 2 C \eps' n
\end{align*}
by~\eqref{algn:trueEq}.
Combining these bounds into~\eqref{eq:whdiff},
\[
  0 \geq C(\weight(T^*)-\weight(\wh{T})) - \frac{1}{5}(C^2 + C) \eps n
\]
or
\[
  \kl{P}{P_{\wh{T}}} = \weight(T^*)-\weight(\wh{T}) \leq \frac{1}{5}(C + 1) \eps n \leq \eps n.\qedhere
\]
\end{proof}

\ignore{
\begin{proof}

Let $G$ be the true tree and $H$ be an MST w.r.t. empirical mutual informations with $N$ large enough. Let $\weight(\cdot)$ and $\hat{\weight}(\cdot)$ denote the true and estimated sum of mutual informations. Let $P$ be a true distribution on $G$ and $Q$ be a best distribution on $H$. Let $\weight(S)$ denote the sum of the true mutual informations on any set of edges $S$. We assume Theorem~\ref{thm:conditionalindependencetesting} succeeds, which happens with at least $1-\delta$ probability.
 We need the following fact.

\begin{fact}\label{fact:swapTree}
Let $T_1$ and $T_2$ be two spanning trees on $n$ vertices such that their symmetric difference consists of the edges $E=\{e_1,e_2\dots,e_{\ell}\}\in T_1\setminus T_2$ and $F=\{f_1,f_2\dots,f_{\ell}\}\in T_2\setminus T_1$. Then $E$ and $F$ can be paired up, say $\langle e_i,f_i \rangle$ WLOG, such that
\begin{itemize}
\item[-] $T_1 \cup \{f_i\}$ and $T_2 \cup \{e_i\}$ have a cycle
\item[-] $T_1 \cup \{f_i\} \setminus \{e_i\}$ and $T_2 \cup \{e_i\} \setminus \{f_i\}$ do not have a cycle
\end{itemize}
\end{fact}
\begin{proof}
We put an edge $f^*\in F$ to $T_1$ which must create a cycle in $T_1$. $f^*$ is a bridge in $T_2$: let $L$ and $R$ be its two connected components. The cycle in $T_1$ mentioned earlier, visualized in $T_2$, must cross $L$ to $R$ at some edge $e^*\notin T_2$. Then, $e^*$ and $f^*$ together must be part of a cycle in $T_2$ as well, since $L$ and $R$ are connected components. Let's call this cycle $C^*$.

Hence $\langle e^*,f^*\rangle=\langle e_1,f_1\rangle$ WLOG is a valid pairing. Let $T_3$ be the tree $T_1 \cup \{f^*\} \setminus \{e^*\}$. Then $T_3$ and $T_2$ differ in $(\ell-1)$ edges and hence can be paired up inductively, such that $e_i$ and $f_i$ together form a cycle in both $T_3$ and $T_2$ for every $i\ge 2$. Let $C_i$ be the cycle due to $\langle e_i,f_i \rangle$ in $T_3$. We claim that $e_i$ and $f_i$ together also form a cycle in $T_1$. The remaining analysis is split into a few cases. See Figure~\ref{fig:treeSwap} for illustrations.
\begin{figure}[h]\label{fig:treeSwap}
\centering
\caption{Figure for Fact~\ref{fact:swapTree}, to be tikzed}
\end{figure}

Case 1,2: If $C^*$ does not involve $e_i$, we are done. Since in that case, we can use $C^*\setminus \{f^*\}$ to complete the cycle, when $f^*$ is removed and $e^*$ is put back. 

Cases 3, 4: These cases create a problem if we replace the edge $f^*$ by the path $C^*\setminus \{f^*\}$, since after such replacement $\langle e_i,f_i \rangle$ no longer form a cycle. To prevent these cases, we make sure if there is an
$f^*$ whose cycle with $e^*$ do not involve any other edge from $E$, that $\langle f^*,e^* \rangle $ pair is chosen first. Then in both Cases 3 and 4, we would always pick the pair $\langle f_i,e^*\rangle$ in the beginning, such that once $f_i$ is added and $e^*$ is removed, $\langle e_i,f^*\rangle$ continues to be a valid pairing in $T_3$.
\end{proof}

Let $\{\langle e_i,f_i \rangle\}_i$ be a pairing given by Fact~\ref{fact:swapTree} for $G$ and $H$.
Theorem~\ref{thm:diffTree} gives us $\dkl(P,Q)=(\weight(G)-\weight(H))=\sum_i (\weight(\{e_i\})-\weight(\{f_i\}))$. We claim that $(\weight(\{e_i\})-\weight(\{f_i\})) \le 2\eps/n$ for every $i$, whence the theorem follows.

Let $e_i=(X,Y)$ be the edge and $f_i=(W,Z)$ be the non-edge in $G$, which together form a cycle in $G$ from Fact~\ref{fact:swapTree}. Then similar to Claim~\ref{claim:four} we get for the true distribution, 
 \[I(X; Z \mid Y) = I(Z; W \mid X) = 0\]
and
\[
    I(X; Y) - I(W; Z) = I(X; Y \mid Z) + I(Z; X \mid W)
  \]

For the sake of contradiction, assume the choice $(W, Z)$ over $(X, Y)$ is more than $2\eps/n$-bad. Then at
  least one of $I(X; Y \mid Z)$ and $I(Z; W \mid W)$ must be more than
  $\eps/n$.  The corresponding empirical conditional mutual
  information is then at least $C\eps/n\abs{\Sigma}$ from Theorem~\ref{thm:conditionalindependencetesting}. Theorem~\ref{thm:conditionalindependencetesting} also gives us $\max(I(\wh{X}; \wh{Z} \mid \wh{Y}),I(\wh{Z}; \wh{W} \mid \wh{X})) \le \alpha C\eps/N\abs{\Sigma}$.

  The empirical distribution will have
  \begin{align}\label{eqn:empTree}
    I(\wh{X}; \wh{Y}) - I(\wh{W}; \wh{Z})
    &= I(\wh{X}; \wh{Y} \mid
      \wh{Z}) - I(\wh{X}; \wh{Z} \mid \wh{Y}) + I(\wh{Z}; \wh{X} \mid
      \wh{W}) - I(\wh{Z}; \wh{W} \mid \wh{X})\nonumber \\
    &\geq C\eps/N\abs{\Sigma} -2\alpha C\eps/N\abs{\Sigma} = (1-2\alpha)C\eps/N\abs{\Sigma} > 0
  \end{align}
Note that $f_i$ is an edge and $e_i$ is a non-edge in $H$, which together form  a cycle in $H$ from Fact~\ref{fact:swapTree}. Hence removing $f_i$ and adding $e_i$ instead give us a strictly larger spanning tree from \eqref{eqn:empTree}, contradicting the maximum spanning tree property of $H$ w.r.t. empirical mutual information.
\end{proof} 
}

\section{Distribution Recovery}\label{sec:distribution}

This section shows how, for a fixed tree $T$, to find a $T$-structured
distribution $Q$ with $\kl{P}{Q} \leq \kl{P}{P_T} + \eps$.  We start
by analyzing how to learn an arbitrary distribution over $\Sigma$.


\subsection{$\dkl$ Learning of Discrete Distributions}
Given $N$ samples from a distribution $P$ over $\Sigma$, the
``add-1'' empirical estimator is based on Laplace's rule of
succession.  This distribution $Q$ is defined by: for each item $i \in \Sigma$, if $i$
appears $T_i$ times in the samples, then
$Q_i={T_i+1\over N+\abs{\Sigma}}$. Kamath, Orlitsky, Pichapati and Suresh~\cite{pmlr-v40-Kamath15} have analyzed the expected behavior of
the add-1 empirical estimator. In this section, we analyze its
behavior in the high-probability regime.
\begin{theorem}\label{thm:KLlearn}
  Let $P$ be a distribution over $\Sigma$ and $N \geq 1$.  Let $Q$ be
  the empirical add-1 estimator from $N$ samples of $P$.  There is
  an universal constant $C>0$ such that, with probability $1-\delta$,
  \[
    \kl{P}{Q}\le  \frac{C\abs{\Sigma} \log {\abs{\Sigma}\over \delta} \log N}{N}.
  \]
\end{theorem}

\begin{proof}
  Let $C' > 1$ be a large constant to be determined later.  If
  $N \leq C' \abs{\Sigma}$, the result follows from
  $\kl{P}{Q} \leq \log \frac{1}{\min_i Q_i} \leq \log (N +
  \abs{\Sigma}) \lesssim \log \abs{\Sigma}$, so we may assume
  $N \geq C' \abs{\Sigma}$. Then

\begin{align*}
\kl{P}{Q}&=\sum_i P_i \log {P_i\over Q_i}\\
&=\sum_i f(P_i-Q_i,Q_i)&&\tag{where $f(x,a)=a[(1+{x\over a})\log (1+{x\over a})-{x\over a}]$}\\
&\lesssim\sum_i  \min \left( {(P_i-Q_i)^2\over Q_i},\abs{P_i-Q_i}\log \left(1+{\abs{P_i-Q_i}\over Q_i}\right)\right)&&\tag{From Lemma~\ref{lem:fapprox}}\\
&\lesssim\sum_i  \min \left( {(P_i-Q_i)^2\over Q_i},\abs{P_i-Q_i}\log N\right).&&\tag{Since $Q_i\ge {1\over N+\abs{\Sigma}}$}.
\end{align*}

We also know from Claim~\ref{claim:Pempirical} that with probability at least $1-\delta$ for each $i$,
\begin{align*}
\abs{P_i-{T_i\over N}} &\lesssim \sqrt{{P_i\log {1\over \delta}\over N}}+{\log {1\over \delta}\over N}\\
\implies \abs{P_i-Q_i} &\lesssim \sqrt{{P_i\log {1\over \delta}\over N}}+{\log {1\over \delta}\over N}+\left\lvert{T_i\over N}-{T_i+1\over N+\abs{\Sigma}}\right\rvert\\
 &= \sqrt{{P_i\log {1\over \delta}\over N}}+{\log {1\over \delta}\over N}+{\abs{\abs{\Sigma}T_i/N-1}\over N+\abs{\Sigma}}\\
&\lesssim \sqrt{{P_i\log {1\over \delta}\over N}}+{\log {1\over \delta}\over N}+{\abs{\Sigma}\over N+\abs{\Sigma}}{T_i\over N}\\
&\lesssim \sqrt{{P_i\log {1\over \delta}\over N}}+{\log {1\over \delta}\over N}\\&\qquad+{\abs{\Sigma}\over N+\abs{\Sigma}}\left(P_i+\sqrt{{P_i\log {1\over \delta}\over N}}+{\log {1\over \delta}\over N}\right)\\
&\lesssim \sqrt{{P_i\log {1\over \delta}\over N}}+{\log {1\over \delta}\over N}+{\abs{\Sigma}\over N+\abs{\Sigma}}P_i
\end{align*}

If $P_i\le {C'\log {1\over \delta}\over N}$, then
$\abs{P_i-Q_i}\log N \lesssim {\log {1\over \delta}\over N}\log N$.

If $P_i> {C'\log {1\over \delta}\over N}$, then
$\abs{P_i - Q_i} \lesssim P_i(\frac{1}{\sqrt{C'}} + \frac{1}{C'} +
\frac{1}{C'+1})$ is at most $\frac{P_i}{2}$ for sufficiently large
$C'$, so $Q_i\ge P_i/2$ and hence
${(P_i-Q_i)^2\over Q_i} \lesssim {\log {1\over \delta}\over
  N}+P_i\left(\abs{\Sigma}\over N+\abs{\Sigma}\right)^2$.

By a union bound, with probability at least $1-\abs{\Sigma}\delta$ we
have that these equations hold for all $i$.  If true, then
\[
  \kl{P}{Q}\lesssim \frac{\abs{\Sigma}}{N}\log {1\over
    \delta}\log N+\left(\abs{\Sigma}\over N+\abs{\Sigma}\right)^2
\lesssim \frac{\abs{\Sigma}}{N}\log {1\over
    \delta}\log N
\]
Rescaling $\delta$ gives the desired bound.
\end{proof}

\subsection{Learning Trees}
We are ready to prove the main result of this section. Algorithm~\ref{algo:algofixed} shows the algorithm we analyze below.
\begin{algorithm}
\SetKwInOut{Input}{input}
\SetKwInOut{Output}{output}
\SetAlgoLined
\caption{Learning closest $T$-structured distribution}
\label{algo:algofixed}
\Input{Samples access to $P$ over $\Sigma^n$; Rooted tree $T$ on $n$ nodes labeled in topological order.}
\Output{$n$ row-stochastic $|\Sigma| \times |\Sigma|$ matrices $Q_1, \dots, Q_n$ that define a $T$-structured distribution $Q$ by 
$Q(x) = \prod_{i=1}^n Q_i[x_{\text{pa}(i)}, x_i]$ (where $x_{\text{pa}(1)}$ is arbitrary).}
\BlankLine
Draw $N$ i.i.d.~samples $X^{(1)}, \dots, X^{(N)}$ from $P$\;
\For{$i \leftarrow 1$ \KwTo $n$}{
\For{$x \in \Sigma$}{
$k \leftarrow \sum_{j=1}^N \mathbf{1}\left[X^{(j)}_{\text{pa}(i)} = x\right]$\tcp*{condition on parent satisfied vacuously if $i=1$} 
\For{$y \in \Sigma$}{
$t \leftarrow \sum_{j=1}^N \mathbf{1}\left[X^{(j)}_{\text{pa}(i)} = x, X^{(j)}_{i} = y\right]$\tcp*{condition on parent sat.~vacuously if $i=1$} 
$Q_i[x, y] \leftarrow (t+1)/(k+|\Sigma|)$\;
}
}
}
\KwRet{$(Q_1, \dots, Q_n)$}
\end{algorithm}
\fixedlearning*
\begin{proof}
Note that $\kl{P}{Q}$ is bounded, because $Q(x)>0$ for all $x$.
  By Lemma~\ref{lem:ApproxChowLiu} and \eqref{eq:ChowLiuEquality}, the learned $T$-structured
  distribution $Q$ satisfies
  \begin{align}\label{eq:KLDiff}
    &\kl{P}{Q} - \kl{P}{P_T}\nonumber\\ &= \sum_{i \in [n]} \sum_{x \in \Sigma} \Pr[X_{\text{pa}(i)} = x]\cdot \kl{X_i \mid X_{\text{pa}(i)} = x}{X'_i \mid X'_{\text{pa}(i)} = x}
  \end{align}
  where $X \sim P$ and $X' \sim Q$.  Now, the node-wise conditional probabilities of $Q$ are the
  add-1 empirical distribution of the conditional probabilities of $P$.  Therefore by
  Theorem~\ref{thm:KLlearn}, if we have $k$ samples of
  $(X \mid X_{\text{pa}(i)}=x)$, we will have with probability
  $1-\delta$ that
  \begin{align*}
    \kl{X_i \mid X_{\text{pa}(i)} = x}{X'_i \mid X'_{\text{pa}(i)} = x} &\lesssim \frac{\abs{\Sigma} \log(\abs{\Sigma}/\delta) \log k}{k}\\ &\leq \frac{\abs{\Sigma} \log(\abs{\Sigma}/\delta) \log N}{k}.
  \end{align*}
  If $\E[k] = N \Pr[X_{\text{pa}(i)} = x] >15\log \frac{1}{\delta}$\arnab{I don't like the $>O(\cdot)$ notation. I put a constant here.},
  then by a Chernoff bound, with probability $1-\delta$ we have
  $k > \frac{1}{2} \E[k]$ and
  \begin{align}
    &\Pr[X_{\text{pa}(i)} = x]\cdot\kl{X_i \mid X_{\text{pa}(i)} = x}{X'_i \mid X'_{\text{pa}(i)} = x} \nonumber\\&\lesssim \frac{\abs{\Sigma} \log(\abs{\Sigma}/\delta) \log N}{N}
 \label{eq:XD}
  \end{align}
  On the other hand, if $\E[k]\leq 15\log \frac{1}{\delta}$, then
  $\kl{X_i \mid X_{\text{pa}(i)} = x}{X'_i \mid X'_{\text{pa}(i)} = x}
  \leq \log (k + \abs{\Sigma})$ [because $Q$ is the add-1
  estimator, so the minimum probability is $\frac{1}{k + \abs{\Sigma}}$]
  and we have
  \begin{align*}
    &\Pr[X_{\text{pa}(i)} = x]\cdot\kl{X_i \mid X_{\text{pa}(i)} = x}{X'_i \mid X'_{\text{pa}(i)} = x}\\ &\leq \frac{\E[k]}{N} \log (k + \abs{\Sigma}) \nonumber\\&\lesssim \frac{\log \frac{1}{\delta} \log N}{N}
  \end{align*}
  Regardless, each term in~\eqref{eq:KLDiff} is bounded by~\eqref{eq:XD}.  Taking a
  union bound, with probability $1 - n\abs{\Sigma} \delta$ we have
  \[
    \kl{P}{Q} - \kl{P}{P_T} \lesssim n\abs{\Sigma} \frac{\abs{\Sigma} \log(\abs{\Sigma}/\delta) \log N}{N}.
  \]
  Rescaling $\delta$ and choosing $N$ appropriately gives the result.
\end{proof}
The algorithm and analysis in Theorem \ref{thm:fixed} straightforwardly generalizes to Bayes nets. This shows that if $G$ is a directed acyclic graph with in-degree bounded by $d$, we can obtain a $G$-structured distribution $Q$ using $\wt{O}(n|\Sigma|^{d+1}/\eps)$ samples from $P$ which satisfies $\kl{P}{Q}-\kl{P}{P_G} \leq \eps$, where $P_G = \arg \min_{G\text{-struct. }R} \kl{P}{R}$. 

\section{Lower Bounds for Structure Recovery}\label{sec:lower}

\subsection{Non-Realizable Case}

This section focuses on the non-realizable case, i.e., the input distribution is not necessarily a tree structured distribution. We prove that $\Omega(n^2/\eps^2)$ samples from a distribution $P$ over $\{0,1\}^n$ are required to find an $\eps$-approximate tree for $P$. First, we prove a lower bound for $n=3$. We define three distributions $R_1, R_2, R_3$ each over $\{0,1\}^3$ as follows. 

Let $B_1\sim \rm{Ber}(1/2)$, $B_2\sim \rm{Ber}(1/2)$, and $B_3\sim \rm{Ber}(1/2)$ be three i.i.d. random bits. In $R_1$: $X_1,Y_1,Z_1$ copies $B_1$ with probabilities ${3\over 4}+\eps,{3\over 4}+\eps,{3\over 4}-\eps$ respectively and with the remaining probabilities, they independently follow $\rm{Ber}(1/2)$. In $R_2$: $X_2,Y_2,Z_2$ copies $B_2$ with probabilities ${3\over 4}+\eps,{3\over 4}-\eps,{3\over 4}+\eps$ respectively and with the remaining probabilities, they independently follow $\rm{Ber}(1/2)$. In $R_3$: $X_3,Y_3,Z_3$ copies $B_3$ with probabilities ${3\over 4}-\eps,{3\over 4}+\eps,{3\over 4}+\eps$ respectively and with the remaining probabilities, they independently follow $\rm{Ber}(1/2)$. We need the following Fact to get our lower bound for $n=3$.

\begin{fact}\label{fact:kllb3}
\emph{(i)} $ \kl{R_1}{R_2} = O(\eps^2)$ and \emph{(ii)} $ I(X_1;Y_1)-I(X_1;Z_1) \ge 0.4\eps$.
\begin{proof}
The calculations are skipped.
\end{proof}
\end{fact}

\begin{lemma}\label{lem:lb3non}
\begin{enumerate}
\item[(i)] Every tree $T$ on 3 vertices is not $0.4\eps$-approx. for one of $R_1,R_2,R_3$.
\item[(ii)] Given samples from a uniformly random distribution among $R_1,R_2,R_3$, $\Omega(1/\eps^2)$ samples are needed to rule out an incorrect one with probability at least 4/5.
\end{enumerate}
\end{lemma}
Observe that Lemma \ref{lem:lb3non} implies that if the distribution that generates the samples is chosen uniformly at random among $R_1,R_2,R_3$, then any algorithm that outputs an $\eps$-approximate tree with error probability $<1/5$ must draw $\Omega(\eps^{-2})$ samples. 

\begin{proof}[Proof of Lemma \ref{lem:lb3non}]
For part (i), note that the possible trees are $G$: $X$---$Z$---$Y$, $H$: $X$---$Y$---$Z$ and $F$: $Y$---$X$---$Z$. Let $R_{1,G}$ denote the closest $G$-structured distribution to $R$ and so on. Then using Corollary \ref{cor:CLweights} and Fact~\ref{fact:kllb3}, $\kl{R}{R_{1,G}}-\kl{R}{R_{1,H}}=I(X_1;Y_1)-I(X_1;Z_1)\ge 0.4\eps$, which means $\kl{R}{R_{1,G}}\ge 0.4\epsilon$. Similarly for $R_2$ and $R_3$ by symmetric calculations.

The proof for part (ii) is skipped.
\end{proof}

We can now prove the main result of this section.
\begin{theorem}\label{thm:nrlb}
The sample complexity of computing an $\eps$-approx. tree on $n$ variables with error probability less than $1/3$ is $\Omega(n^2 \eps^{-2})$.
\end{theorem}
\begin{proof}
Suppose $n=3m$ is a multiple of $3$. We consider the $n$ variables as being divided into $m$ blocks, each of size $3$. Let $P$ be a random distribution on $\{0,1\}^n$, defined by setting the distribution of the $i$'th block to be $R$ or $R'$ with probability 1/2 each independently, where $R$ and $R'$ satisfy Lemma \ref{lem:lb3non}.

For the sake of contradiction, suppose we have an algorithm that draws $cn^2\eps^{-2}$ samples from $P$ (for a sufficiently small constant $c$) and outputs an $\eps$-approximate tree $T$ with probability at least $2/3$ (over the choice of $P$ as well as the algorithm's randomness). Since each block is independent, without loss of generality, $T$ is a union of disjoint trees $T_1, \dots, T_m$ for each block. By Lemma \ref{lem:lb3non} (with $\eps$ replaced with $10\eps/m$), each $T_i$ is not $10\eps/m$-approximate with probability at least $1/5$. Hence, by a Chernoff bound, with probability $>2/3$, for at least $\frac{m}{10}$ trees, $T_i$ is not $10\eps/m$-approximate. Therefore, for any $T$-structured distribution $Q$, $\kl{P}{Q} > \frac{m}{10}\cdot \frac{10\eps}{m} = \eps$.  
\end{proof}

\subsection{Realizable Case}
We now show that if $P$ is a tree-structured distribution on $n$ variables, then $\Omega(n\eps^{-1}\log {n\over \eps})$ samples are required to find an $\eps$-approx. tree. As with the non-realizable case, we first show the construction for $n=3$. We define three distributions $R_1,R_2,R_3$ each over $\{0,1\}^3$ as follows.

In $R_1$, $(Y_1,Z_1)$ randomly takes values between $(0,0)$ and $(1,1)$, and  $X_1$ copies the other 2 bits with $(1-\eps)$ probability and with the remaining probability follows $\mathrm{Ber}(1/2)$. $R_2$ and $R_3$ are defined symmetrically with the restrictions $X_2=Z_2$ and $X_3=Y_3$ respectively. Let $H(P,Q)=\sqrt{{1\over 2}\sum_{x\in \Omega} \left(\sqrt{P(x)}-\sqrt{Q(x)}\right)^2}$ be the Hellinger distance between two distributions. We need the following facts. 

\begin{fact}\label{fact:real3lb}
\emph{(i)} $H^2(R_1,R_2)=\Theta(\eps)$ and \emph{(ii)} $I(Y_1;Z_1)-I(X_1;Z_1)=\Theta(\eps\log {1\over \eps})$.
\end{fact}
\begin{proof}
The calculations are skipped.
\end{proof}

\begin{lemma}\label{lem:rlb3}
\begin{enumerate}
\item[(i)] Every tree $T$ on $3$ vertices is not $\Theta(\eps\log {1\over \eps})$-approximate for one of $R_1,R_2,R_3$.
\item[(ii)] Given samples from a random distribution from $R_1,R_2,R_3$, $\Omega(1/\eps)$ samples are needed to rule out an incorrect one with 4/5 probability.
\end{enumerate}
\end{lemma}
The same argument used for Theorem \ref{thm:nrlb} implies:
\begin{theorem}\label{thm:rlb}
The sample complexity of computing an $\eps$-approx. tree for a tree-structured distribution on $n$ variables with error probability less than $1/3$ is $\Omega(n\eps^{-1}\log {n\over \eps})$. 
\end{theorem}
\begin{proof}[Proof of Lemma \ref{lem:rlb3}]
For part (i), let $G:=$ $X$---$Y$---$Z$ and $H:=$ $Y$---$X$---$Z$. Let $R_{1,H}$ be the closest $H$-structured distribution to $R_1$ and so on. Then by Corollary~\ref{cor:CLweights}, $\kl{R_1}{R_{1,H}}-\kl{R_1}{R_{1,G}}=I(Y_1;Z_1)-I(X_1;Z_1)=\Omega(\eps\log {1\over \eps})$ using Fact~\ref{fact:real3lb}. Hence $\kl{R_1}{R_{1,H}}\allowbreak=\Omega(\eps\log {1\over \eps})$. Similarly, it can be shown for the other two trees.

The proof for part (ii) is skipped. 
\end{proof}


\paragraph{Acknowledgement} Arnab and Sutanu thank Rishi Gajjala for contributing to related discussions and Sanjoy Dasgupta for clarifying some questions about \cite{DBLP:journals/ml/Dasgupta97}. Also, Eric thanks Alex Dimakis for pointing us to \cite{DBLP:journals/corr/abs-2010-14864}. We thank the anonymous reviewers of STOC '21 for valuable comments. Arnab was supported in part by NRF-AI Fellowship R-252-100-B13-281 and an Amazon Research Award. Sutanu was supported in part by NRF-AI Fellowship R-252-100-B13-281. Eric was supported in part by NSF Award CCF-1751040 (CAREER). Vinod was supported in part by NSF CCF-184908 and NSF HDR:TRIPODS-1934884 awards.

\appendix
\section{Proofs of Background}\label{app:ApproxChowLiu}

\begin{proof}[Proof of Lemma~\ref{lem:ApproxChowLiu}]
  We have that
  \begin{align*}
    &\sum_{v \in V} I(X_v; X_{\text{pa}(v)}) = \sum_{v \in V} \left(H(X_v) - H(X_v \mid X_{\text{pa}(v)})\right)=\\
    &\sum_{v \in V} H(X_v) + \sum_{x\in\Sigma^n} \Pr[X = x] \sum_{v \in V}  \log \Pr[X_{v} = x_{v} \mid X_{\text{pa}(v)} = x_{\text{pa}(v)}].
  \end{align*}
  Therefore
  \begin{align*}
    &\kl{P}{Q}\\ &= \sum_{x\in\Sigma^n} \Pr[X=x] \log \frac{\Pr[X=x]}{\Pr[X'=x]}\\
              &= -H(X) + \sum_{x\in\Sigma^n} \Pr[X=x] \log \frac{1}{\Pr[X'=x]}\\
              &= -H(X) + \sum_{x\in\Sigma^n} \Pr[X=x] \sum_{v \in V} \log \frac{1}{\Pr[X'_v=x_v \mid X'_{\text{pa}(v)} = x_{\text{pa}(v)}]}\\
              &= -H(X) + \sum_{v \in V} H(X_v) - \sum_{v \in V} I(X_v; X_{\text{pa}(v)})\\
                &\qquad + \sum_{x\in\Sigma^n} \Pr[X=x] \sum_{v \in V} \log \frac{\Pr[X_v=x_v \mid X_{\text{pa}(v)} = x_{\text{pa}(v)}]}{\Pr[X'_v=x_v \mid X'_{\text{pa}(v)} = x_{\text{pa}(v)}]}\\
              &= -H(X) + \sum_{v \in V} H(X_v) - \sum_{v \in V} I(X_v; X_{\text{pa}(v)})\\
                & + \sum_{v \in V}\sum_{x\in\Sigma^2} \Pr[(X_{\text{pa}(v)},X_v) = x]  \log \frac{\Pr[X_v=x_v \mid X_{\text{pa}(v)} = x_{\text{pa}(v)}]}{\Pr[X'_v=x_v \mid X'_{\text{pa}(v)} = x_{\text{pa}(v)}]}
  \end{align*}
  which is the desired bound.
\end{proof}


\section{Proof of Fact~\ref{fact:kllb3}}\label{sec:micalc}
\emph{Part (i):} The following table computes the p.m.f. of $R_1$.
\begin{center}
\renewcommand{\arraystretch}{2}
 \begin{tabular}{||c c c c c||} 
 \hline
 $X_1$ & $Y_1$ & $Z_1$ & $B$ & value \\ [0.5ex] 
 \hline
 \hline
 0 & 0 & 0 & 0 & ${1\over 2}\left({7\over 8}+{\eps\over 2}\right)^2\left({7\over 8}-{\eps\over 2}\right)$\\
 &&&1&$\left({1\over 8}-{\eps\over 2}\right)^2\left({1\over 8}+{\eps\over 2}\right)$\\
 0 & 0 & 1 & 0 & ${1\over 2}\left({7\over 8}+{\eps\over 2}\right)^2\left({1\over 8}+{\eps\over 2}\right)$\\
 &&&1 & ${1\over 2}\left({1\over 8}-{\eps\over 2}\right)^2\left({7\over 8}-{\eps\over 2}\right)$\\
 
0 & 1 & 0 & 0 & ${1\over 2}\left({7\over 8}+{\eps\over 2}\right)\left({1\over 8}-{\eps\over 2}\right)\left({7\over 8}+{\eps\over 2}\right)$\\
&&&1&${1\over 2}\left({1\over 8}-{\eps\over 2}\right)\left({7\over 8}+{\eps\over 2}\right)\left({1\over 8}+{\eps\over 2}\right)$\\

0 & 1 & 1 & 0 & ${1\over 2}\left({7\over 8}+{\eps\over 2}\right)\left({1\over 8}-{\eps\over 2}\right)\left({1\over 8}+{\eps\over 2}\right)$\\

 &  &  & 1 & ${1\over 2}\left({1\over 8}-{\eps\over 2}\right)\left({7\over 8}+{\eps\over 2}\right)\left({7\over 8}-{\eps\over 2}\right)$

\\\hline
\end{tabular}
\end{center}

We skipped entries with $X_1=1$ in the above table which can be evaluated using $R_i(x,y,z)=R_i(1-x,1-y,1-z)$. From the above table, $R_1(0,0,1)={1\over 2}\left({7\over 64}+{3\eps\over 8}+{3\eps^2\over 4}\right)$ and $R_1(0,1,0)={1\over 2}\left({7\over 64}-{3\eps\over 8}-{\eps^2\over 4}\right)$. Note that by definition, $R_2(x,y,z)=R_1(x,z,y)$. Hence, we get $\kl{R_1}{R_2}=2\left[ R_1(0,0,1)\log {R_1(0,0,1)\over R_1(0,1,0)}+R_1(0,1,0) \log {R_1(0,1,0)\over R_1(0,0,1)}\right]=O(\eps^2)$.

\emph{Part (ii):} We compute the marginal distribution on $X_1,Y_1$ from the previous table.
\ignore{
\begin{center}
\renewcommand{\arraystretch}{2}
 \begin{tabular}{||c c c c||} 
 \hline
 $X_1$ & $Y_1$ & $B$ & value \\ [0.5ex] 
 \hline
 \hline
 0 & 0 & if $B=0$ & ${1\over 2}\left[ \left( \left(0.75+\eps\right)+\left(0.25-\eps\right){1\over 2}\right) \left( \left(0.75+\eps\right)+\left(0.25-\eps\right){1\over 2}\right)\right]$\\ \hline

 && if $B=1$ & ${1\over 2}\left[ (0.25-\eps){1\over 2}(0.25-\eps){1\over 2} \right]$\\ \hline

 1 & 1 & if $B=0$ & ${1\over 2}\left[ (0.25-\eps){1\over 2}(0.25-\eps){1\over 2} \right]$\\ \hline

  &  & if $B=1$ & ${1\over 2}\left[ \left( \left(0.75+\eps\right)+\left(0.25-\eps\right){1\over 2}\right) \left( \left(0.75+\eps\right)+\left(0.25-\eps\right){1\over 2}\right)\right]$\\ \hline

0 & 1 & if $B=0$ & ${1\over 2}\left[ \left( (0.75+\eps)+(0.25-\eps){1\over 2} \right)(0.25-\eps){1\over 2} \right]$\\ \hline

 &  & if $B=1$ & ${1\over 2}\left[ (0.25-\eps){1\over 2}\left( (0.75+\eps)+(0.25-\eps){1\over 2} \right) \right]$\\ \hline

1 & 0 & if $B=0$ & ${1\over 2}\left[ (0.25-\eps){1\over 2}\left( (0.75+\eps)+(0.25-\eps){1\over 2} \right) \right]$\\ \hline

 &  & if $B=1$ & ${1\over 2}\left[ \left( (0.75+\eps)+(0.25-\eps){1\over 2} \right)(0.25-\eps){1\over 2} \right]$\\ \hline

\end{tabular}
\end{center}

Simplifying, we get:
\begin{center}
\renewcommand{\arraystretch}{2}
 \begin{tabular}{|| c c c||} 
 \hline
 $X_1$ & $Y_1$ & value \\ [0.5ex] 
 \hline
 \hline
 0 & 0 & ${1\over 2}\left[ \left({7\over 8}+{\eps\over 2}\right)^2+\left({1\over 8}-{\eps\over 2}\right)^2\right]$\\ \hline
 1 & 1 & ${1\over 2}\left[ \left({7\over 8}+{\eps\over 2}\right)^2+\left({1\over 8}-{\eps\over 2}\right)^2\right]$\\ \hline
 0 & 1 & ${1\over 2}\left[ 2 \left({7\over 8}+{\eps\over 2}\right)\left({1\over 8}-{\eps\over 2}\right)\right ]$ \\ \hline
 1 & 0 & ${1\over 2}\left[ 2 \left({7\over 8}+{\eps\over 2}\right)\left({1\over 8}-{\eps\over 2}\right)\right ]$ \\ \hline
\end{tabular}
\end{center}
}
$
R_1(0,0,\dot)=R_1(1,1,\cot)={1\over 2}\left[ \left({7\over 8}+{\eps\over 2}\right)^2+\left({1\over 8}-{\eps\over 2}\right)^2\right]={25\over 64}+{\eps^2\over 4}+{3\eps\over 8} \approx {25\over 64}\left(1+{24\eps\over 25}\right)$ and $R_1(0,1,\dot)=R_1(1,0,\dot)=
\left({7\over 8}+{\eps\over 2}\right)\left({1\over 8}-{\eps\over 2}\right) = {7\over 64}-{3\eps\over 8}-{\eps^2\over 4}\approx {7\over 64}\left(1-{24\eps\over 7}\right)$

Hence,
\begin{flalign*}
&H(X_1,Y_1)\nonumber\\&= 2\left[ {25\over 64} \left( 1+{24\eps\over 25}\right)\log {64\over 25\left(1+{24\eps\over 25}\right)}+{7\over 64} \left( 1-{24\eps\over 7}\right)\log {64\over 7\left(1-{24\eps\over 7}\right)}\right] \\\nonumber
&\approx 2\left[ {25\over 64}\left(1+{24\eps\over 25}\right)\left( \log {64\over 25}-\log\left(1+{24\eps\over 25}\right) \right)+ \right.\\&\qquad\qquad\qquad\qquad\qquad \left.{7\over 64}\left(1-{24\eps\over 7}\right)\left( \log {64\over 7}-\log\left(1-{24\eps\over 7}\right) \right)\right]&&\tag{using $\log (1+x)\approx x$}\\\nonumber
\end{flalign*}
\begin{align}
H(X_1,Y_1)&\approx C_{12}\text{ (a constant term corresponding to $\eps=0$}) - 0.48\eps \pm O(\eps^2)\label{eqn:R1R2}
\end{align}

Similarly, we compute the following marginal distribution on $X_1,Z_1$.\ignore{
\begin{center}
\renewcommand{\arraystretch}{2}
 \begin{tabular}{||c c c c||} 
 \hline
 $X_1$ & $Z_1$ & $B$ & value \\ [0.5ex] 
 \hline
 \hline
 0 & 0 & if $B=0$ & ${1\over 2}\left[ \left((0.75+\eps)+(0.25-\eps){1\over 2}\right)\left((0.75-\eps)+(0.25+\eps){1\over 2}\right) \right]$\\ \hline
  &  & if $B=1$ & ${1\over 2}\left[ (0.25-\eps){1\over 2}(0.25+\eps){1\over 2} \right]$\\ \hline
 
 1 & 1 & if $B=0$ & ${1\over 2}\left[ (0.25-\eps){1\over 2}(0.25+\eps){1\over 2} \right]$\\ \hline
   &  & if $B=1$ & ${1\over 2}\left[ \left((0.75+\eps)+(0.25-\eps){1\over 2}\right)\left((0.75-\eps)+(0.25+\eps){1\over 2}\right) \right]$\\ \hline
   
 1 & 0 & if $B=0$ & ${1\over 2}\left[(0.25-\eps){1\over 2}\left( (0.75-\eps)+(0.25+\eps){1\over 2}\right)\right]$ \\ \hline
 & & if $B=1$ & ${1\over 2}\left[ \left( (0.75+\eps)+(0.25-\eps){1\over 2}\right)(0.25+\eps){1\over 2} \right]$ \\\hline
 
 0 & 1 & if $B=0$ & ${1\over 2}\left[ \left( (0.75+\eps)+(0.25-\eps){1\over 2}\right)(0.25+\eps){1\over 2} \right]$ \\\hline
 & & if $B=1$ & ${1\over 2}\left[(0.25-\eps){1\over 2}\left( (0.75-\eps)+(0.25+\eps){1\over 2}\right)\right]$ \\\hline 
\end{tabular}
\end{center}

Simplifying, we get:
\begin{center}
\renewcommand{\arraystretch}{2}
 \begin{tabular}{|| c c c||} 
 \hline
 $X_1$ & $Z_1$ & value \\ [0.5ex] 
 \hline
 \hline
 0 & 0 & ${25\over 64}-{\eps^2\over 4}$\\ \hline
 1 & 1 & ${25\over 64}-{\eps^2\over 4}$\\ \hline
 0 & 1 &  ${7\over 64}+{\eps^2\over 4}$\\ \hline
 1 & 0 &  ${7\over 64}+{\eps^2\over 4}$\\ \hline
\end{tabular}
\end{center}
} $R_1(0,\dot,0)=R_1(1,\dot,1)={25\over 64}-{\eps^2\over 4}$ and $R_1(0,\dot,1)=R_1(1,\dot,0)={7\over 64}+{\eps^2\over 4}$ 
Hence,
\begin{align}\label{eqn:R1R3}
H(X_1,Z_1) &= C_{13}\text{ (a constant term corresponding to $\eps=0$}) \pm O(\eps^2)
\end{align}

Finally, we have from \eqref{eqn:R1R2} and \eqref{eqn:R1R3},
\begin{align*}
    I(X_1;Y_1)-I(X_1;Z_1) &= H(X_1,Z_1)-H(X_1,Y_1) &&\tag{since the marginal entropies are 0 due to unbiasedness}\\
    &\approx 0.48\eps\pm(\eps^2) &&\tag{using $C_{12}=C_{13}$, since the three bits are identically distributed when $\eps=0$}
\end{align*}

\section{Proof of Fact~\ref{fact:real3lb}}
\emph{Part (i)} We first compute the p.m.f. of $R_1$: $R_1(0,0,0)={1\over 2}\left(1-{\eps\over 2}\right)$, $R_1(0,0,1)=R_1(0,1,0)=0$, and $R_1(0,1,1)={1\over 2}{\eps\over 2}$.
\ignore{
\begin{center}
\renewcommand{\arraystretch}{2}
\begin{tabular}{||c c c c||}
\hline
$X_1$ & $Y_1$ & $Z_1$ & value \\ [0.5ex]
\hline\hline
0 & 0 & 0 & ${1\over 2}\left(1-{\eps\over 2}\right)$\\\hline
0 & 0 & 1 & 0\\\hline
0 & 1 & 0 & 0\\\hline
0 & 1 & 1 & ${1\over 2}{\eps\over 2}$\\\hline
\end{tabular}
\end{center}
}
Other entries of the p.m.f. can be filled using $R_i(x,y,z)=R_i(1-x,1-y,1-z)$. The p.m.f. of $R_2$ can be computed by observing $R_2(x,y,z)=R_1(y,x,z)$. Then, $H^2(R_1,R_2)={\eps\over 2}$.

\emph{Part (ii)}
\begin{align*}
I(Y_1;Z_1)-I(X_1;Z_1)&=H(X_1;Z_1)-H(Y_1;Z_1)\\
&=\left[\left(1-{\eps\over 2}\right)\log {1\over {1\over 2}\left(1-{\eps\over 2}\right)}+{\eps\over 2}\log {1\over {1\over 2}{\eps\over 2}}
\right]-
\left[
\log 2
\right] \\
&=O(\eps)+{\eps\over 2}\log {2\over \eps}
\end{align*}
\section{Proof of Lemma~\ref{lem:lb3non} and Lemma~\ref{lem:rlb3} part (ii)}\label{sec:3pt}
We show the following result, which can be proven by generalizing Le Cam's two point method.

\begin{claim}\label{claim:3pt}
Let $R_1,R_2,R_3$ be three distributions over $\Omega$ such that $\max\{|R_1-R_2|_{TV},|R_1-R_3|_{TV}\}\le \eps$. Suppose nature randomly chooses a distribution $R_i$ among $R_1,R_2,R_3$ and reveal exactly one sample from it. Then for any algorithm which tries to guess a $j$ such that $R_j\neq R_i$, must err with probability at least ${1\over 3}(1-2\eps)$.
\end{claim}
\begin{proof}
Suppose the algorithm upon seeing the sample $x\in \Omega$, decides to output $R_1,R_2,R_3$ with probabilities $C_{x,1},C_{x,2},C_{x,3}$ respectively such that $C_{x,1}+C_{x,2}+C_{x,3}=1$. Then the error probability:
\begin{align*}
&\mathrm{Pr}[R_i=R_j]\\&=\mathrm{Pr}[j=1\text{ and }i=1]+\mathrm{Pr}[j=1\text{ and }i=1]+\mathrm{Pr}[j=1\text{ and }i=1]\\
&={1\over 3}\left(\mathrm{Pr}[j=1\mid i=1]+\mathrm{Pr}[j=1\mid i=1]+\mathrm{Pr}[j=1\mid i=1]\right)\\
\end{align*}
\begin{align*}
&\mathrm{Pr}[j=1\mid i=1]\\&=\sum_{x\in \Omega} \mathrm{Pr}[j=1\text{ and sample=}x \mid i=1]\\
&=\sum_{x\in \Omega} \mathrm{Pr}[j=1\mid\text{ sample=}x\text{ and }i=1]\cdot \mathrm{Pr}[\text{sample=}x\mid i=1]\\
&=\sum_{x\in \Omega} \mathrm{Pr}[j=1\mid\text{ sample=}x]\cdot \mathrm{Pr}[\text{sample=}x\mid i=1]\\
&=\sum_{x\in \Omega} C_{x,1}R_1(x)\\
\end{align*}
\begin{align*}
&\mathrm{Pr}[R_i=R_j]\\&={1\over 3}\left(
\sum_{x\in \Omega} C_{x,1}R_1(x)+
\sum_{x\in \Omega} C_{x,2}R_2(x)+
\sum_{x\in \Omega} C_{x,3}R_3(x)
\right)\\
&={1\over 3}\left(1+\sum_x C_{x,2}(R_2(x)-R_1(x))+\sum_x C_{x,3}(R_3(x)-R_1(x))\right)&&\tag{Using $C_{x,1}+C_{x,2}+C_{x,3}=1$.}\\
&\ge {1\over 3}\left(1-|R_1-R_2|_{TV}-|R_1-R_3|_{TV}\right)
\end{align*}
The last inequality follows from the definition of TV distance, which is the minimum value of the previous expression when $C_{x,2}$ indicates $R_2(x)<R_1(x)$ and $C_{x,3}$ indicates $R_3(x)<R_1(x)$.  
\end{proof}

We obtain part (ii) of Lemma~\ref{lem:lb3non} from Claim~\ref{claim:3pt} by using $k$-fold distributions of $R_1,R_2,R_3$. These distributions have pairwise KL distance $O(k\eps^2)$ from Fact~\ref{fact:kllb3} and hence if $k=o(1/\eps^2)$ then $TV=o(1)$ from Pinsker's inequality.

Similarly, we obtain part (ii) of Lemma~\ref{lem:rlb3}. We have $H^2(R_1^{\otimes k},R_2^{\otimes k})=1-(1-H^2(R_1,R_2))^k$ using Hellinger factorization for product distributions. Since  $H^2(R_1,R_2)=O(\eps)$, we get if $k=o(1/\eps)$ then $|R_1^{\otimes k}-R_2^{\otimes k})|_{TV}\lesssim H(R_1^{\otimes k},R_2^{\otimes k})=o(1)$.

\bibliographystyle{alpha}
\bibliography{reflist}

\end{document}